\documentclass[a4paper, 11pt]{article}
\usepackage{xcolor}
\usepackage{xifthen}
\usepackage{algorithmic}
\usepackage{comment}
\usepackage{tcolorbox}
\usepackage{relsize}
\usepackage[font=footnotesize]{caption}
\usepackage[normalem]{ulem}
\usepackage{bm}
\usepackage{soul}
\usepackage[algo2e,ruled]{algorithm2e}

\def\fontsettingup{2} 
\usepackage{amsmath, amsthm, amssymb}
\usepackage[T1]{fontenc}
\ifthenelse{\fontsettingup = 1}{ \usepackage{eulerpx, eucal, tgpagella}   }{}
\ifthenelse{\fontsettingup = 2}{ \usepackage{mathpazo, tgpagella} }{}
\usepackage{hyperref}
\hypersetup{
  colorlinks=true,
  linkcolor=blue,
  filecolor=magenta,
  urlcolor=cyan,
  citecolor=blue
}
 \DontPrintSemicolon
   \SetKwInOut{Input}{Input}
   \SetKwInOut{Output}{Output}
\usepackage{cleveref}
\usepackage{makecell}

\usepackage[a4paper]{geometry}
\newgeometry{
textheight=9in,
textwidth=6.5in,
top=1in,
headheight=14pt,
headsep=25pt,
footskip=30pt
}
\newtheorem{theorem}{Theorem}%[section]

\newtheorem*{claim*}{Claim}

\newtheorem{fact}[theorem]{Fact}
\newtheorem{lemma}[theorem]{Lemma}

\newtheorem{corollary}[theorem]{Corollary}
\theoremstyle{definition}

\newtheorem{definition}[theorem]{Definition}
\newtheorem{remark}[theorem]{Remark}
\newtheorem*{remark*}{Remark}

\ifthenelse{\fontsettingup = 1}{
  \def\*#1{\mathbf{#1}} % Use \*A for \mathbf{A}
  \def\+#1{\mathcal{#1}} % Use \+A for \mathcal{A}
  \def\-#1{\mathrm{#1}} % Use \-A for \mathrm{A}
  \def\^#1{\mathbb{#1}} % Use \^A for \mathbb{A}
  \def\!#1{\mathfrak{#1}} % Use \!A for \mathfrak{A}
}{}

\ifthenelse{\fontsettingup = 2}{
  \def\*#1{\boldsymbol{#1}} % Use \*A for \mathbf{A}
  \def\+#1{\mathcal{#1}} % Use \+A for \mathcal{A}
  \def\-#1{\mathrm{#1}} % Use \-A for \mathrm{A}
  \def\^#1{\mathbb{#1}} % Use \^A for \mathbb{A}
  \def\!#1{\mathfrak{#1}} % Use \!A for \mathfrak{A}
}{}

\def\oPr{\mathbf{Pr}}
\renewcommand{\Pr}[2][]{ \ifthenelse{\isempty{#1}}
  {\oPr\left[#2\right]}
  {\oPr_{#1}\left[#2\right]} } % Use \Pr[a]{b} for \mathbf{Pr}_a[b], \Pr{b} for  \mathbf{Pr}[b]

\def\oE{\mathbb{E}}
\newcommand{\E}[2][]{ \ifthenelse{\isempty{#1}}
  {\oE\left[#2\right]}
  {\oE_{#1}\left[#2\right]} }

\DeclareMathOperator*{\oVar}{\mathbf{Var}}
\newcommand{\Var}[2][]{ \ifthenelse{\isempty{#1}}
  {\oVar\left[#2\right]}
  {\oVar_{#1}\left[#2\right]} }

\def\oEnt{\mathbf{Ent}}
\newcommand{\Ent}[2][]{ \ifthenelse{\isempty{#1}}
  {\oEnt\left[#2\right]}
  {\oEnt_{#1}\left[#2\right]} }
\renewcommand{\epsilon}{\varepsilon}
\renewcommand{\emptyset}{\varnothing}

\newcommand{\paren}[1]{(#1)}
\newcommand{\Paren}[1]{\left(#1\right)}

\let\epsilon=\varepsilon

\title{A Generalized Binary Tree Mechanism for Differentially Private Approximation of All-Pair Distances}
\author{
Michael Dinitz\thanks{Department of Computer Science, Johns Hopkins University, USA. \textnormal{E-mail: \url{mdinitz@cs.jhu.edu}}. Supported in part by NSF awards CCF-1909111 and CCF-2228995.}
\and
Chenglin Fan\thanks{Department of Computer Science and Engineering, Seoul National University, Republic of Korea. \textnormal{E-mail: \url{fanchenglin@gmail.com}}.}
\and
Jingcheng Liu\thanks{State Key Laboratory for Novel Software Technology and New Cornerstone Science Laboratory, Nanjing University, China. \textnormal{E-mail: \url{liu@nju.edu.cn}, \url{zou.zongrui@smail.nju.edu.cn}}. JL and ZZ have been supported by National Science Foundation of China under Grant No. 62472212 and the New Cornerstone Science Foundation.}
\and 
Jalaj Upadhyay\thanks{Management Science \& Information Systems Department, Rutgers University, USA. \textnormal{E-mail: \url{jalaj.upadhyay@rutgers.edu}}. Research supported by the Rutgers Decanal Grant no. 302918 and an unrestricted gift from Google.}
\and 
Zongrui Zou\footnotemark[3]
}
\date{}

\begin{document}
\pagenumbering{arabic}
\maketitle

\begin{abstract}
% In this paper, we generalize the classic binary tree mechanism from paths or trees to a broader class of graphs, which we refer to as \textit{recursively separable} graphs, in the context of computing all-pairs shortest distances with differential privacy. As applications, we present an efficient algorithm for privately approximating the all-pairs shortest-path distances in \( n \)-vertex $K_h$-minor-free graphs with a purely additive error of \( \widetilde{O}(h(Wn)^{1/3} ) \), where \( W \) represents the maximum edge weight. For grid graphs, using the same scheme, we further reduce the additive error to \( \widetilde{O}(n^{1/4} W) \).
We study the problem of approximating all-pair distances in a weighted undirected graph with differential privacy, introduced by Sealfon [Sea16]. Given a publicly known undirected graph, we treat the weights of edges as sensitive information, and two graphs are neighbors if their edge weights differ in one edge by at most one. We obtain efficient algorithms with significantly improved bounds on a broad class of graphs which we refer to as \textit{recursively separable}. In particular, for any \( n \)-vertex $K_h$-minor-free graph, our algorithm achieve an additive error of \( \widetilde{O}(h(nW)^{1/3} ) \), where \( W \) represents the maximum edge weight; For grid graphs, the same algorithmic scheme achieve additive error of \( \widetilde{O}(n^{1/4}\sqrt{W}) \).

Our approach can be seen as a generalization of the celebrated binary tree mechanism for range queries, as releasing range queries is equivalent to computing all-pair distances on a path graph. In essence, our approach is based on generalizing the binary tree mechanism to graphs that are \textit{recursively separable}.

\end{abstract}

\section{Introduction}
{
Graph is one of the central data-structure that is used to encode variety of datasets. As a result, analysis on graph finds applications in wide areas in not just computer science (networking, search engine optimization, bioinformatics, machine learning, etc), but also in biology, chemistry, transportation and logistics, supply-chain management, operation research, and even in everyday applications like finding the shortest route on a map using GPS navigation. Given these graphs encodes sensitive information, the analysis of sensitive \emph{graph} data while preserving privacy has attracted growing attention in recent years. One of the main privacy notions that has been proposed and used extensively in such graph analysis is {\em differential privacy}, with successful applications including cut approximation \cite{gupta2012iterative, eliavs2020differentially, DBLP:conf/nips/DalirrooyfardMN23,liu2024optimal, chandra2024differentially}, spectral approximation \cite{blocki2012johnson, DBLP:conf/nips/AroraU19, upadhyay2021differentially, liu2024almost}, correlation or hierarchical clustering (\cite{bun2021differentially, DBLP:conf/nips/Cohen-AddadFLMN22, DBLP:conf/icml/ImolaEMCM23, cohen2022scalable}) and numerical statistics release (\cite{DBLP:conf/tcc/KasiviswanathanNRS13, DBLP:conf/nips/UllmanS19, DBLP:conf/focs/BorgsCSZ18, ding2021differentially, imola2022differentially}), among others.

Differential privacy~\cite{dwork2006calibrating} is a rigorous mathematical framework designed to provide strong privacy guarantees for individuals in a dataset. Parameterized by privacy parameter $(\epsilon,\delta)$, differential privacy is a property on the algorithm instead of other proposed notions, like $k$-anonymity, that can be seen as a property of the dataset. Formally, an algorithm $\mathcal{A}$ is $(\epsilon,\delta)$-\emph{differentially private} if the probabilities of obtaining any set of possible outputs $S$ of $\mathcal{A}$ when run on two ``neighboring'' inputs $G$ and $G'$ are similar:
$\Pr{\mathcal{A}(G) \in S} \leq e^{\epsilon} \cdot \Pr{\mathcal{A}(G')\in S} + \delta$. When $\delta = 0$, we say that $\mathcal{A}$ preserves \emph{pure differential privacy}, otherwise $\mathcal{A}$ is \emph{approximately differentially private}. In this definition, the notion of neighboring plays a crucial role in defining what is considered ``private."
}

As the size and complexity of graph data continue to grow, the need for efficient algorithms that can compute shortest-path distances while ensuring privacy has become paramount. Traditional approaches, such as Dijkstra's and Floyd-Warshall algorithms, provide accurate results but clearly lack the privacy guarantees that are necessary for sensitive applications. Recently, started with Sealfon (\cite{sealfon2016shortest}), the problem of releasing \emph{All-Pair Shortest Distances} (APSD) on a weighted graph with differential privacy are receiving increasing attention~\cite{DBLP:conf/nips/Fan0L22, fan2022distances, DBLP:conf/soda/ChenG0MNNX23, ebrahimi2023differentially, DBLP:conf/wads/DengGUW23, bodwin2024discrepancy}, due to their practical applicability and the fact that the problem itself is inherently natural. In this line of works, given an $n$-vertex weighted graph $G = ([n],E,w)$, the goal is to output the values of the shortest distances between each pair of vertices, while preserving \emph{weight}-level differential privacy: i.e., two graphs are considered neighboring if they share the same edge set (i.e., topology) and differ in the weight of exactly one edge by at most one.

Sealfon (\cite{sealfon2016shortest}) applies the post-processing scheme and gives an $\widetilde{O}(n/\epsilon)$ purely additive error on answering all-pair shortest distances in $n$-vertex graphs for pure and approximate differential privacy with privacy budget $\epsilon$. This bound was improved by~\cite{DBLP:conf/nips/Fan0L22, DBLP:conf/soda/ChenG0MNNX23} using a simple sub-sampling technique, reducing it to only $\widetilde{O}(\sqrt{n}/\epsilon)$ for approximate differential privacy. For the lower bound side, the first result on the private APSD is provided by~\cite{DBLP:conf/soda/ChenG0MNNX23}, which establishes a lower bound of $\Omega(n^{1/6}/\epsilon)$. This bound was recently improved to $\Omega(n^{1/4}/\epsilon)$ by~\cite{bodwin2024discrepancy}.

Though $\widetilde{O}(\sqrt{n})$ error is best known for answering private APSD in general graphs, it is possible to significantly improve it for many special classes of graphs. For example, outputting all pairs shortest distances on an $n$-vertex \emph{path} graph is equivalent to answering all possible range queries on a vector of length $n-1$~\cite{sealfon2016shortest}. Under the weight-level privacy, this task can be addressed within only $\text{polylog}(n)$ additive error using the classic \emph{binary tree} mechanism~\cite{dwork2010differential, chan2011private}. This mechanism leverages the Bentley \& Saxe's data structure~\cite{bentley1978decomposable} to reduce the number of compositions required for privacy. More generally, \cite{sealfon2016shortest} shows that the $\text{polylog}(n)$ error can also be achieved when privately answering the shortest distances on an $n$-vertex tree. 
 
However, there is a lack of studies on private APSD in the context between trees and general graphs, a gap that encompasses many types of graphs commonly encountered in practical applications, such as \emph{planar graphs}. Perhaps, the most natural example of such application is finding the shortest distances in road maps or metro networks, where all the routes can usually be represented on a plane, and edge weights are determined by traffic (higher edge weights indicate more traffic, which navigation tools should tend to avoid). Therefore, we have the following natural question:

\begin{quote}
    \textbf{Question 1.} \emph{Is it possible to achieve better utility in privately answering all-pairs shortest distances for planar graphs (or other natural class of graphs) compared to general graphs}\footnote{On graphs with bounded tree-width $p$, an interesting result by \cite{ebrahimi2023differentially} presents a purely differentially private algorithm that outputs all-pair shortest distances with an additive error of $\widetilde{O}(p^2/\epsilon)$. This bound however, does not lead to an improvement on planar graphs, as the tree-width of a planar graph can be as large as $O(\sqrt{n})$.}\emph{?}
\end{quote}

In this paper, we show that for any $K_h$-minor-free graph, there exists an efficient and differentially private algorithm that achieves an additive error of $\widetilde{O}(h\cdot (nW)^{1/3})$ for answering APSD, where $W$ is the maximum edge weights. This result immediately implies an $\widetilde{O}((nW)^{1/3})$ error for private APSD on planar graphs, as every planar graph is $K_5$-minor-free~\cite{kuratowski1930probleme}. Additionally, for a $\sqrt{n}\times \sqrt{n}$ grid graph, within the same framework, we further reduce the error of private APSD to only $\widetilde{O}(\frac{n^{1/4}W^{1/2}}{{\epsilon}^{1/2}})$. These results improve multiple previous upper bounds that also depend on $W$. (See Section \ref{sec:results} for details.) {Such an upper bound $W$ on edge weights is also natural in real-world scenarios, such as packet routing, where it is defined by the bandwidth of the network.} 

From a technical point of view, our results can be considered as a generalization of the mechanism proposed in \cite{sealfon2016shortest} for computing APSD on an $n$-vertex tree. Their mechanism inherently extends the binary tree mechanism, which originally operates on a single path, to a tree structure, achieving an error bound of $\text{polylog}(n)$. Therefore, to further extend this approach to more general graph topologies beyond paths or trees, we pose the following fundamental question:

\begin{quote}
    \textbf{Question 2} \emph{ What underlying property allows the binary tree mechanism to reduce the error in answering APSD queries on paths or trees?}
\end{quote}
To answer this question, we propose the concept of \emph{recursive separability} (Definition \ref{def:recur_separability}). We show that the recursive separability of tree graphs is the inherent reason why the binary tree mechanism significantly reduces the error in answering APSD. %for trees lies in the fact that tree graphs are essentially \emph{recursively separable}. 
Intuitively, a graph as \emph{recursively separable} if it has a \emph{balanced separator} such that when the separator is removed, the resulting subgraphs themselves are also recursively separable. 

By working with recursive separable graphs, we develop a private algorithm based on a divide-and-conquer strategy to compute all-pair shortest distances, with an additive error characterized by the quality of the separators of the input graph. As a consequence, our algorithm not only encompasses the older use cases of the binary tree mechanism~\cite{dwork2010differential, chan2011private, dwork2015pure, sealfon2016shortest}, %\liuexp{ do you mean use cases beyond range queries? please add citation} \zongrui{PTAL.}
but also extends to a much broader class of graphs that exhibit recursive separability, including planar graphs or $K_h$-minor-free graphs. Importantly, our results also open up the opportunity to leverage a rich body of prior work on graph separation theorems (e.g.~\cite{lipton1979separator, gilbert1984separator, alon1994planar, Grigni2004WellConnected, alon1990separator, kawarabayashi2010separator, fox2010separator, korhonen2024induced}) to the design of private APSD algorithms.

\subsection{Our Results}\label{sec:results}
In this paper, we consider the standard \emph{weight-level privacy}~\cite{sealfon2016shortest}. In this notion, two positive weighted $n$ vertices graphs $G = ([n],E,w)$, $G' = ([n],E,w')$ with the same edge topology are \textit{neighboring} if $\|w-w'\|_0 \leq 1$ and $\|w-w'\|_1\leq 1$, i.e., there is an edge which differ by at most $1$. Here, $w,w'\in \mathbb{R}_{\geq 0}^{|E|}$ encodes the edge weights, where $\mathbb R_{\geq 0}$ denote the set of positive real numbers. 

For a graph $G$, a \textit{separator} of $G$ is a subset of vertices that, when removed, splits the graph into two or more disconnected components. Formally, a separator is  a subset of vertices $S\subseteq V$ if there exists disjoint $A,B$ such that $A\cup B = V\backslash S$ and no edge joins a vertices in $A$ with a vertex in $B$. {To privately compute all-pairs shortest distances using a divide-and-conquer approach, we first introduce the concept of recursive separability.}

\begin{definition}\label{def:recur_separability}
Fix $p\in \mathbb{N}$ and $q,q'\in \mathbb{R}$ such that $\frac{1}{2}\leq q<q'<1$. An undirected graph $G = (V,E)$ is \textit{$(p,q,q')$-recursively separable} if and only if either $|V| = O(1)$ or:
\begin{itemize}
    \item $G$ has a separator $S\subseteq V$ of size $|S|\leq \min\{p, (q'-q)|V|\}$ such that each connected component of $G' = (V\backslash S, E(V\backslash S))$ has at most $q|V|$ vertices.
    \item Any subgraph of $G$ is $(p,q,q')$-recursively separable.
\end{itemize}
\end{definition}
\begin{remark}
{
    We note that the definition of recursive separability is {\em oblivious} to edge weights. Thus, we say a weighted graph is recursively separable if its underlying graph is recursively separable.}
\end{remark}
Our first result is a general algorithm that privately computes all-pairs shortest distances for any graph, with an error bound that depends on the quality of the graph's separators.
\begin{theorem}\label{thm:intro1}
    Fix privacy budgets $0<\epsilon,\delta<1$, and suppose $\frac{1}{2}\leq q<q'<1$ are constants. For any weighted $n$-vertex $(p,q,q')$-recursively separable graph $G$, there exists an $(\epsilon,\delta)$-differentially private algorithm such that with high probability, it outputs APSD on $G$ with worst case additive error at most $O\left(\frac{p\cdot \log^3(n/\delta)}{\epsilon}\right)$. 
\end{theorem}
\noindent Notably, Theorem \ref{thm:intro1} offers a wide range of intriguing implications, as the definition of recursive separability captures a large class of natural graphs. Some example with implication includes
\begin{enumerate}
    \item All $n$-vertex trees are $(1,\frac23, \frac23+o(1))$-recursively separable. Therefore, Theorem \ref{thm:intro1} immediately recovers the error bound for privately computing APSD on trees in \cite{sealfon2016shortest} (though with an extra logarithmic factor).
    \item Any graph with tree-width $p$ is $(p+1,\frac23, \frac23+o(1))$-recursively separable (which follows as an immediate corollary of \cite{robertson1986graph}). As a consequence, Theorem \ref{thm:intro1} improves the private APSD result for bounded tree-width graphs in \cite{ebrahimi2023differentially} from $\widetilde{O}(p^2/\epsilon)$ to $\widetilde{O}({p/\epsilon})$\footnote{{This was also suggested as an open problem by Sealfon in personal communication.}}.
\end{enumerate}
%For example, it is straightforward to verify from the definition that a More generally, a

Leveraging the well-known separation theorem for planar graphs~\cite{lipton1979separator, gilbert1984separator}, we also conclude that all planar graphs are $(O(\sqrt{n}), \frac23, \frac23+o(1))$-recursively separable. Though applying Theorem \ref{thm:intro1} does not immediately yield any non-trivial improvement for planar graphs,  we demonstrate that with a slight modification to the recursive framework of the algorithm in Theorem \ref{thm:intro1}, we can achieve an improvement not only on planar graphs but also on $K_h$-minor-free graphs. {Recall that a graph is $K_h$-minor-free if the clique of size $h$, denoted by $K_h$, is not a minor of it.} 
If a graph is $K_{h+1}$-minor-free but contains a $K_h$ minor, then $h$ is also known as the \textit{Hadwiger number}~\cite{hadwiger1943klassifikation}. 
\begin{theorem}\label{thm:intro2}
    Fix privacy budgets $0<\epsilon,\delta<1$, and integers $W, h\geq 1$. For any weighted $K_h$-minor-free graph $G = ([n],E,w)$ with $\|w\|_\infty \leq W$, there exists a $(\epsilon,\delta)$-differentially private algorithm such that with high probability, it outputs APSD on $G$ with worst case additive error at most $\widetilde{O}_\delta\left(h\cdot \frac{(nW)^{1/3}}{\epsilon^{2/3}}\right)$. 
\end{theorem}

By Kuratowski's Theorem \cite{kuratowski1930probleme}, planar graphs are $K_5$-minor-free. Then, we immediately have the following corollary:

\begin{corollary}
Fix privacy budgets $0<\epsilon,\delta<1$, and integer $W\geq 1$. For any planar graph $G = ([n],E,w)$ with $\|w\|_\infty \leq W$, there exists a $(\epsilon,\delta)$-differentially private algorithm such that with high probability, it outputs APSD on $G$ with worst case additive error at most $\widetilde{O}_\delta\left(\frac{(nW)^{1/3}}{\epsilon^{2/3}}\right)$.
\end{corollary}

Building on the same algorithmic framework, we show further improvement for a specific subclass of planar graphs: grid graphs. 

\begin{theorem}\label{thm:intro3}
Fix privacy budgets $0<\epsilon,\delta<1$, and integers $a, W\geq 1$. For any $a\times n/a$ grid graph $G = ([{a}\times {n/a}],E,w)$ with $\|w\|_\infty \leq W$, there exists a $(\epsilon,\delta)$-differentially private algorithm such that with high probability, it outputs APSD on $G$ with worst case additive error at most $\widetilde{O}_\delta\left(\frac{(nW^2)^{1/4}}{\epsilon^{1/2}}\right)$. 
\end{theorem}

In Table \ref{tab:results}, we provide a brief overview of the existing results on differentially private all-pairs shortest distances approximation when assuming an upper bound $W$ on edge weights (in this table, $\alpha = \sqrt{2}-1 \approx 0.4143$). This assumption also appears in \cite{sealfon2016shortest} and \cite{DBLP:conf/soda/ChenG0MNNX23}. %{Notably, when $W=O(1)$, this is the first optimal bound on the shortest distance for any non-trivial class of graphs using the discrepancy lower bounds by Bodwin et al.~\cite{bodwin2024discrepancy}.}

\begin{table*}[t]
    \centering
    \caption{Current error bounds for APSD under $(\epsilon,\delta)$-differential privacy with bounded edge weights.}\label{tab:results}
    \begin{tabular}{|c|c|c|c|}
        
        \hline
         &  \makecell[c]{Sealfon~\cite{sealfon2016shortest}} &  \makecell[c]{Chen et al.~\cite{DBLP:conf/soda/ChenG0MNNX23}} & \makecell[c]{Ours}\\
  \hline
%          \textbf{general graph} &  $\widetilde{O}_\delta\left(\sqrt{\frac{nW}{\epsilon}}\right)$ &  $\widetilde{O}_\delta\left(\frac{W^\alpha n^{\alpha+o(1)}}{\epsilon^{1-\alpha}} + \frac{n^{\alpha+o(1)}}{\epsilon}\right)$ & - \\
 % \hline
            \textbf{planar graph} &  $\widetilde{O}_\delta\left(\sqrt{\frac{nW}{\epsilon}}\right)$ & $\widetilde{O}_\delta\left(\frac{W^\alpha n^{\alpha+o(1)}}{\epsilon^{1-\alpha}} + \frac{n^{\alpha+o(1)}}{\epsilon}\right)$ &  $\widetilde{O}_\delta\left({\frac{(nW)^{1/3}}{\epsilon^{2/3}}}\right)$\\
  \hline
              \textbf{grid graph} &  $\widetilde{O}_\delta\left(\frac{n^{1/3}}{\epsilon} + W\cdot n^{1/3}\right)$ &   (same as above) &  $\widetilde{O}_\delta\left(\frac{n^{1/4}\sqrt{W}}{\epsilon^{1/2}}\right)$\\
  \hline
    \end{tabular}
  \end{table*}

\subsection{Technical Overview}
Here, we introduce the technical ingredients that underpin our results on the private APSD approximation and discuss the key concepts behind our analysis of privacy and utility guarantees.

\subsubsection{Recursively computing APSD by separators}\label{sec:overview_recursive}

One of the most basic approaches to compute all-pair shortest distances on a graph is to add noise to the weight of each edge in $E$. The resulting distances are then private due to the post-processing property of differential privacy. Further, the error on each path is roughly proportional to the hops of the path, which leads to an $\widetilde{O}(n/\epsilon)$ error~\cite{sealfon2016shortest}. Using basic probability theory, it can be shown that with high probability, uniformly sampling $O(\sqrt{n}\log n)$ vertices from the vertex set of a connected graph can form a $k$-covering set {(see Definition \ref{def:covering_set})} of the original graph, where $k = O(\sqrt{n})$. By privately computing all-pair shortest distances within the $O(\sqrt{n})$ vertices in the covering set, the error in private APSD can be further reduced to $\widetilde{O}(\sqrt{n}/\epsilon)$ (\cite{DBLP:conf/nips/Fan0L22,DBLP:conf/soda/ChenG0MNNX23}) since in each path, the covering set is guaranteed to be encountered after passing through the $O(\sqrt{n})$ vertices.

However, the above methods do not make use of any combinatorial properties or topological structure of the graph. To take advantage of this, we observe that some specific classes of graphs, including trees or planar graphs, exhibit good separability, allowing us to use the divide-and-conquer method to recursively compute the shortest path. Specifically, for a graph $G= (V,E,w)$, and $x,y\in V$, let $d_G(x,y)$ be the distance between $x,y$ in $G$. Then, if $G$ has a separator $S\subseteq V$ that separates the graph into two subgraphs $G_1$ and $G_2$, we observe that for any pair $x,y\in V$ that are not in $S$, their distance can be written as $$d_G(x,y) = \mathop{\min}_{a,b\in S} \left\{d_{G_1}(x,a) + d_{G}(a,b) + d_{G_2}(b,y)\right\}$$ if $x\in G_1$ and $y\in G_2$. Then, the error accumulated in each recursion corresponds to the error in privately computing the distances in $G$\footnote{We have to privately compute $d_G(a,b)$ for any $a,b\in S$ instead of $d_S(a,b)$ because that the shortest path between $x,y$ may repeatedly enter and exit the separator $S$.} between each vertex \emph{inside} $S$. Intuitively, the recursion step can be illustrated by a binary tree of $O(\log |V|)$ levels if the input graph can be recursively divided by a series of well-balanced separators. However, this does \emph{not} directly give the $\widetilde{O}(p/\epsilon)$ error if the input graph is $(p,O(1),O(1))$-recursively separable, since the error accumulates across all leaf vertices in the corresponding recursion tree. To address this issue, we employ a \emph{pruning} trick that directly establishes relationships between adjacent levels of separators, preventing errors from accumulating within the same level. In particular, we add shortcuts between adjacent levels of separators to form a series of complete bipartite graphs to control the error.

In summary, the key distinction between our technique and those employed in other works that improves private APSD (\cite{DBLP:conf/nips/Fan0L22, DBLP:conf/soda/ChenG0MNNX23}) lies in our utilization of the graph's topology to construct \emph{highly structured} shortcuts. In contrast, the shortcuts constructed in previous works are oblivious to any combinatorial properties of the graph.

\subsubsection{Approximating distances within a separator using covering sets}

In Section \ref{sec:overview_recursive}, we briefly introduced how to construct an algorithmic framework such that the error in estimating private APSD primarily depends on the size of the separators for recursively separable graphs. Recall that in our recursive framework, we need to construct two types of shortcuts: those within the separators and those between adjacent separators. Therefore, by examining the separability of certain classes of graphs, the overall error for private APSD can be further reduced by optimizing the error incurred when estimating the distances of vertices inside or between separators. In this section, we focus solely on illustrating the idea for reducing the error in estimating the APSD (i.e., the lengths of the shortcuts) inside a separator. The idea for reducing the error in estimating the lengths of shortcuts between two adjacent separators is similar.

Specifically, suppose that the edge weights are upper-bounded by some constant $W\geq 1$. Then, we observe that finding a \emph{$k$-covering set} (Definition \ref{def:covering_set}) of the separator $S$ and computing the all-pair distances for all vertices inside the $k$-covering set provides an approximation of all-pair distances within $S$, with an extra additive error at most $kW$. The advantage is that by computing APSD only for the $k$-covering set, the number of compositions required to preserve privacy is significantly reduced when $k \ll |S|$. This approach is similar to the one in \cite{sealfon2016shortest} for improving the error in estimating APSD for grid graphs. However, without our recursive framework, they use this trick by finding a $k$-covering set for the whole graph instead of only for separators, which limits its ability to achieve further improvements. In our recursive framework, we observe that a grid graph is separable with a series of \emph{connected} separators of size at most $O(\sqrt{n})$. Consequently, one can show that each separator has a $(n^{1/4})$-covering set of size at most $O(n^{1/4})$, which finally leads to an $\widetilde{O}(n^{1/4})$ error (for constant $W$) for grid graphs by the advanced composition theorem.

Achieving improvements on general planar graphs is similar but somewhat more tricky, as the separators are not necessarily connected. To make progress, for a connected planar graph $G$ on $n$ vertices, we first partition the vertex set of $G$ into about $O(n/d)$ disjoint subsets, such that each subset has diameter at most $d$ (in terms of the subgraph defined by this subset of vertices). Then, we contract each subset into a super-node, merge all multi-edges, and obtain a smaller planar graph $\widetilde{G}$ with $O(n/d)$ super-nodes. By the separation theorem for planar graphs~\cite{lipton1979separator}, $\widetilde{G}$ has a separator $\widetilde{S}$ of size at most $O(\sqrt{n/d})$, and thus privately computing the all-pair shortest distances inside $\widetilde{S}$ incurs an error of at most $\widetilde{O}(\sqrt{n/d\epsilon^2})$. This implies an approximation of the APSD within the original separator $S$ with an error of 
$$2dW + \widetilde{O}(\sqrt{n/d\epsilon^2}),$$
since the diameter of each supper-node is at most $d$. By choosing an appropriate value for $d$, we obtain an $\widetilde{O}((nW)^{1/3})$ error for estimating APSD inside separators of a planar graph, and thus the same error with just extra logarithmic factors for estimating the APSD in the entire graph. We note that the same approach can be identically applied to $K_h$-minor-free graphs (see also Lemma~\ref{lem:k-cover_in_separator}). Note that the dependence on edge weights $W$ arises because we need to estimate the all-pair distances inside a subgraph $S$, yet the shortest path between two nodes in $S$ may traverse vertices outside of $S$. Therefore, we need this assumption to ensure that the shortest path within $S$ remains relatively close in length to the global (true) shortest path.

\section{Notations and Preliminaries}

Throughout this paper, we work on weighted graphs $G = (V,E,w)$ where $w\in \mathbb{R}_+^E$ encodes edge weights. Let $a,b$ be two binary strings, we use $a \circ b$ to denote the binary string obtained concatenating $a$ and $b$. When edge weights are not of concern, for a graph $G=(V,E)$ and $S\subseteq V$, we use $G_S = (S,E(S))$ to denote the subgraph induced by $S$. For any $h\geq 1$, we use $K_h$ to denote a clique of size $h$. A graph $H$ is a minor of a graph $G$ if a copy of $H$ can be obtained from $G$ via repeated edge deletion or edge contraction, and we say a graph $G$ is \emph{$K_h$-minor-free} if $G$ does not have $K_h$ as its minor. It is well-known that all planar graphs are $K_5$-minor-free:

\begin{lemma}[The Kuratowski's reduction theorem~\cite{kuratowski1930probleme, wagner1937uber}] A graph $G$ is planar if and only if the complete graph $K_5$ and the complete bipartite graph $K_{3,3}$ are not minors of $G$.
    
\end{lemma}

Here, we provide the necessary background on differential privacy to facilitate understanding of this paper. A key feature of differential privacy algorithms is that they preserve privacy under post-processing. That is to say, without any auxiliary information about the dataset, any analyst cannot compute a function that makes the output less private. 
\begin{lemma}
    [Post processing~\cite{dwork2014algorithmic}] Let $\mathcal{A}:\mathcal{X}\rightarrow \mathcal{R}$ be a $(\epsilon,\delta)$-differentially private algorithm. Let $f:\mathcal{R}\rightarrow \mathcal{R}'$ be any function, then $f\circ \mathcal{A}$ is also $(\epsilon,\delta)$-differentially private.
\end{lemma}

Sometimes we need to repeatedly use differentially private mechanisms on the same dataset, and obtain a series of outputs.
\begin{lemma}[Basic composition~\cite{dwork2006calibrating}]\label{lem:composition}
     let $D$ be a dataset in $\mathcal{X}$ and $\mathcal{A}_1, \mathcal{A}_2,\cdots, \mathcal{A}_k$ be $k$ algorithms where $\mathcal{A}_i$ (for $i\in [k]$) preserves $(\epsilon_i,\delta_i)$ differential privacy, then the composed algorithm $\mathcal{A}(D) = (\mathcal{A}_1(D), \cdots, \mathcal{A}_2(D))$ preserves $(\sum_{i\in [k]}{\epsilon_i}, \sum_{i\in [k]}{\delta_i})$-differential privacy.
\end{lemma}

\begin{lemma}[Advanced composition~\cite{dwork2010boosting}]\label{lem:adv_composition}
    For parameters $\epsilon>0$ and $\delta,\delta'\in [0,1]$, the composition of $k$ $(\epsilon,\delta)$-differentially private algorithms is a $(\epsilon', k\delta+\delta')$ differentially private algorithm, where 
    $$\epsilon' = \sqrt{2k\log(1/\delta')} \cdot \epsilon + k\epsilon (e^\epsilon - 1).$$
\end{lemma}

The Laplace mechanism is one of the most basic mechanisms to preserve differential privacy for numeric queries, which calibrates a random noise from the Laplace distribution (or double exponential distribution) according to the $\ell_1$ sensitivity of the function. 

\begin{lemma}
    (Laplace mechanism)  Suppose $f:\mathcal{X}\rightarrow \mathbb{R}^k$ is a query function with $\ell_1$ sensitivity $\mathsf{sens}_1(f)\leq \Delta_1$. Then the mechanism
    $$\mathcal{M}(D) = f(D) + (Z_1,\cdots,Z_k)^\top$$
    is $(\epsilon,0)$-differentially private, where $Z_1,\cdots, Z_k$ are i.i.d random variables drawn from $\texttt{Lap}(\Delta_1/\epsilon)$.
\end{lemma}

Adding Gaussian noise based on the $\ell_2$ sensitivity guarantees approximate differential privacy.

\begin{lemma}
    (Gaussian mechanism) Suppose $f:\mathcal{X}\rightarrow \mathbb{R}^k$ is a query function with $\ell_2$ sensitivity $\mathsf{sens}_1(f)\leq \Delta_2$. Then the mechanism
    $$\mathcal{M}(D) = f(D) + (Z_1,\cdots,Z_k)^\top$$
    is $(\epsilon,\delta)$-differentially private, where $Z_1,\cdots, Z_k$ are i.i.d random variables drawn from $\mathcal{N}\left(0, \frac{(\Delta_2)^2\cdot 2\ln(1.25/\delta)}{\epsilon}\right)$.
\end{lemma}

\section{Private APSD Approximation For Recursively Separable Graph}\label{sec:recursive}

In this section, we study answering all-pair shortest distances on \emph{recursively separable} graphs with differential privacy. Let $G$ be a graph with non-negative edge weights, and let $x,y$ be two vertices in $G$. we use $d_G(x,y)$ to denote the (shortest) distance between $x$ and $y$ in $G$. Specifically, we establish the following result:
\begin{theorem} [Restatement of Theorem \ref{thm:intro1}]\label{thm:main}
   Fix any $0<\epsilon,\delta<1$, and $n,p\in \mathbb{N}$. For any graph $G = ([n],E,w)$ that is $(p,q,q')$-recursively separable (\Cref{def:recur_separability}) for some constants $\frac{1}{2} \leq q \leq q' <1$, there is an $(\epsilon,\delta)$-differentially private algorithm such that, with high probability, it outputs estimations of shortest distances $\{\widehat{d}_G(x,y)\}_{x,y\in [n]}$ satisfying 
    $$|\widehat{d}_G(s,t) - d_G(s,t)| \leq O\left( \frac{p \cdot \log^2 n\log(n/\delta)}{\epsilon} \right).$$
\end{theorem}

\paragraph{Proof outline of Theorem~\ref{thm:main}.} We provide our algorithms for privately computing APSD in recursively separable graphs in Section~\ref{sec:algorithm}. In Section~\ref{sec:privacy}, we present the privacy guarantee (Theorem~\ref{thm:privacy}) and its analysis of our algorithms. The utility guarantee (Theorem~\ref{thm:utility}) is given in Section~\ref{sec:utility}, where it is proven by induction, with a detailed analysis deferred to Appendix~\ref{app:utility}. Combining Theorem~\ref{thm:privacy} and Theorem~\ref{thm:utility} concludes the proof of Theorem~\ref{thm:main}.

\subsection{The Algorithm}\label{sec:algorithm}

Our algorithm is built upon a sequence of decompositions applied to the input graph, with the resulting process being traceable through a binary tree.

\paragraph{The Construction of the Binary Tree.} Given an unweighted and undirected graph $G=([n],E)$ that is $(p,q,q')$-recursively separable, we consider a deterministic and recursive procedure that finally separates the graphs into $O(n)$ pieces of subgraphs of constant size (we note that the vertex set in these pieces may have intersection). {We will see later that the bounded number of pieces is ensured by the property of recursive separability.} More specifically, in the first epoch, since $G$ is $(p,q,q')$-recursively separable, then there exists an $S\subseteq [n]$ such that removing vertices in $S$ resulting two subgraphs {$G_0' = (V_0',E(V_0')), G_1' =(V_1',E(V_1'))$} that are not connected, and {$\max\{|V_0'|, |V_1'|\}\leq qn.$} We let 
\begin{align}
    \label{eq:firstrecursivestep}
    G_0 = (V_0' \cup S, E(V_0'\cup S) \backslash E(S))\quad \text{and} \quad G_1 = (V_1' \cup S, E(V_1'\cup S) \backslash E(S)).
\end{align} 

That is, we union $S$ into $G_1'$ (or $G_2'$) together with all edges incident between $S$ and $G_1'$ (or $G_2'$). From the recursive separability of $G$, we also have that $$\max\{|V_0|, |V_1|\}\leq qn + |S| \leq q'n,$$
where $V_0$ (or $V_1$) is the set of vertices of $G_0$ (or $G_1$).

In the recursive subroutine, for graph $G_b$ where $b$ is a binary string of length that depends on the current depth of recurrence, there must exist a separator $S_b$ that separates $G_b$ into $G_{b\circ 1}'$ and $G_{b\circ 0}'$. Similarly as in \cref{eq:firstrecursivestep}, we let $$G_{b\circ 1} = (V_{b\circ 1}' \cup S_b, E(V_{b\circ 1}' \cup S_b) \backslash E(S_b)) \quad \text{ and } \quad G_{b\circ 0} =(V_{b\circ 0}' \cup S_b, E(V_{b\circ 0}' \cup S_b) \backslash E(S_b)).$$ 

The recursive procedure terminates when $G_b$ has constant size $c$. Clearly, the depth of the recursion is $\log_{1/q'} \frac{n}{c} = {\log(n/c) \over \log(1/q')} = O(\log n)$ for any constant $0<q'<1$. Further, these graphs construct a binary tree $\mathcal{T}$ with $O(\log n)$ levels, where the root of $\mathcal{T}$ is $G$, and the nodes in the $i$-th level is just the all $G_b$'s with $|b| = i$, and each $G_b$ in the leaves of $\mathcal{T}$ has size at most $c$. In each non-leaf node of $\mathcal{T}$, we associate it with a label $(G_b, S_b)$ to denote the subgraph it represents and the separator used to further split $G_b$. We note that for the root, $b = \emptyset$.

\paragraph{The Algorithm.}  With the construction of the binary tree described above, we are ready to present our algorithm. For keeping the presentation modular, we present it in form of three algorithms: Algorithm \ref{alg:shortcut} defines how we add shortcuts, Algorithm \ref{alg:recursive-apsp} describes the recursive procedure that is used by Algorithm \ref{alg:apsp} to compute all-pair-shortest distances. 

In Algorithm \ref{alg:shortcut}, we first add highly structured shortcuts based on the decomposition described above to generate a private synthetic graph. {In particular, there are two types of shortcuts: (1) between all pairs of vertices within each separator that is determined by the decomposition procedure, forming a complete graph, and (2) between vertices in every pair of adjacent separators, forming a complete bipartite graph.} 

Then, as a post-processing stage, in Algorithm \ref{alg:apsp} we compute the all-pair shortest path distances privately for a weighted graph by calling the recursive procedure for each pair (Algorithm \ref{alg:recursive-apsp}). Here, for any $s,t\in V$, we use 
$d_b(s,t)$ to denote the value of the \textit{local} shortest distance between $s$ and $t$ in the subgraph $G_b$. If either $s$ or $t$ does not appear in $V_b$, then we set $d_b(s,t) = \infty$. In our algorithms, for any $x,y\in [n]$ and $b\in \{0,1\}^*$, we define $\texttt{IsShortcut}(x,y,b) = \textsf{\em True}$ if a noisy edge (i.e., a shortcut) is added between $x,y$ within the separator of $G_b$, or between the separator of $G_b$ and that of its precursor graph.
%\jalaj{Maybe, give step 4 a name because we change only this step later on}

\begin{algorithm2e}[!ht]
\SetKwInput{Input}{Input}
\SetKwInOut{Output}{Output}
\Input{Graph $G=({V},{E},w)$, private parameter $\epsilon,\delta$.}
\DontPrintSemicolon
1. Recursively construct a binary tree $\mathcal{T}$ as described in this section.\;

2. Set $h = \log_{1/q'}(n/c)$,  $\epsilon' = \epsilon/\sqrt{4h\log(1/\delta')}$, and $\delta' = \delta/(4h)$.\;

3. Let $\sigma = p\sqrt{2\log(1.25/\delta')}/\epsilon'$.\;

4. \For{non-leaf node $(G_b, S_b) \in V(\mathcal{T})$} 
{
\For{$x,y\in S_b$ such that $x\neq y$}{
\texttt{IsShortcut}$(x,y,b) = \textsf{\em True}$.\; 
%\jalaj{Where is IsShortcut defined?}
Let $\widehat{d}_b(x,y) = d_b(x,y) + \mathcal{N}(0,\sigma^2)$.\;
}
\If{$b\neq \emptyset$}{
Let $b'$ be the binary string that removes the last bit in $b$.\;
\For{$(x,y) \in S_{b'}\times S_{b}$}{
\If{$x,y\notin S_{b'} \cap S_b$}{
\texttt{IsShortcut}$(x,y,b') = \textsf{\em True}$.\;
 
Let $\widehat{d}_b(x,y) = d_b(x,y) + \mathcal{N}(0,\sigma^2)$.\;
}
}
}
}
5. \For{leaf node $(G_b,\text{-}) \in V(\mathcal{T})$}{
\For{$x,y\in V_b$ such that $x\neq y$}{
\texttt{IsShortcut}$(x,y,b) = \textsf{\em True}$.\; 
 
Let $\widehat{d}_b(x,y) = d_b(x,y) + \mathcal{N}\left(0,\frac{2c^2\log(1.25/\delta')}{(\epsilon')^2}\right)$.\;
}
}
\Output{The binary tree $\mathcal{T}$, and $\widehat{d}_b(u,v)$ for all $u,v\in V$ and $b\in \{0,1\}^*$ such that $\texttt{IsShortcut}(x,y,b) = \textsf{\em True}$.}
\caption{Constructing private shortcuts.}\label{alg:shortcut}
%\algolab{private_path}
\end{algorithm2e}

\begin{algorithm2e}[!ht]
\SetKwInput{Input}{Input}
\SetKwInOut{Output}{Output}
\Input{A binary tree constructed out of some graph $G=(V,E,w)$, a graph $G_b$, a pair of vertices $s,t\in V_b$, integer $k\in \mathbb{N}$.}

1. \If{$\texttt{IsShortcut}(s,t,b) = \textsf{\em True}$}
{
Halt and directly output $\widehat{d}_b(s,t)$ computed by Algorithm \ref{alg:shortcut}.
}
2. \If{$k=0$}{

\If{both $s,t\in V_{b\circ a}$ for $a\in \{0,1\}$}{
\For{$z\in S_b$}{
$\widehat{d}_{b\circ a}(s,z)\leftarrow$ Recursive-APSD$(\mathcal{T},G_{b\circ a}, (s,z), k+1)$.\;
$\widehat{d}_{b\circ a}(t,z)\leftarrow$ Recursive-APSD$(\mathcal{T},G_{b\circ a}, (t,z), k+1)$.\;
 }
$\widehat{d}_{b\circ a}(s,t) \leftarrow $ Recursive-APSD$(\mathcal{T},G_{b\circ a}, (s,t), 0)$.\;
   $\widehat{d}_b(s,t) \leftarrow \mathop{\min}\{\widehat{d}_{b\circ a}(s,t), \min_{x,y\in S_b} \widehat{d}_{b\circ a}(s,x) + \widehat{d}_b(x,y) + \widehat{d}_{b\circ a}(y,t)\}$\;
}
\If{$s\in V_{b\circ a}$ and $t\in V_{b\circ \bar{a}}$,
}
{
\For{$z\in S_b$}{
$\widehat{d}_{b\circ a}(s,z)\leftarrow$ Recursive-APSD$(\mathcal{T},G_{b\circ a}, (s,z), k+1)$.\;
$\widehat{d}_{b\circ \bar{a}}(t,z)\leftarrow$ Recursive-APSD$(\mathcal{T},G_{b\circ \bar{a}}, (t,z), k+1)$.\;
 }
$\widehat{d}_b(s,t) \leftarrow  \min_{x,y\in S_b} \widehat{d}_{b\circ a}(s,x) + \widehat{d}_b(x,y) + \widehat{d}_{b\circ \bar{a}}(y,t)$\;
}

}
3. \If{$k>0$}{
Let $b'$ be the binary string that removes the last bit in $b$.\;
\If{$s,t \notin S_{b'}$}{
Halt and output \texttt{FAIL}.
}
WLOG let $t\in S_{b'}$ (if not we just switch $s$ and $t$).\;
\For{$z\in S_b$}{
$\widehat{d}_{b\circ a}(s,z)\leftarrow$ Recursive-APSD$(\mathcal{T},G_{b\circ a}, (s,z), k+1)$.\;
 }
\If{both $s,t\in V_{b\circ a}$ for $a\in \{0,1\}$}{
$\widehat{d}_{b\circ a}(s,t) \leftarrow $ Recursive-APSD$(\mathcal{T},G_{b\circ a}, (s,t), 0)$.\;

   Let $\widehat{d}_b(s,t) \leftarrow \mathop{\min}\{\widehat{d}_{b\circ a}(s,t), \min_{x\in S_b} \widehat{d}_{b\circ a}(s,x) + \widehat{d}_b(x,t)$\}.\;
}
\If{$s\in V_{b\circ a}$ and $t\in V_{b\circ \bar{a}}$,
}
{
$\widehat{d}_b(s,t) \leftarrow  \min_{x\in S_b} \widehat{d}_{b\circ a}(s,x) + \widehat{d}_b(x,t)$.\;
}

}

\Output{estimated local distance $\widehat{d}_b(s,t)$.}

\caption{Recursive-APSD$(\mathcal{T}, G_b, (s,t), k)$}
\label{alg:recursive-apsp}
\end{algorithm2e}

\begin{algorithm2e}[!ht]
\SetKwInput{Input}{Input}
\SetKwInOut{Output}{Output}
\Input{Graph $G=({V},{E},w)$, private parameter $\epsilon,\delta$.}
\DontPrintSemicolon
1. Run Algorithm \ref{alg:shortcut} on $G$ and $\epsilon,\delta$, obtain the labeled binary tree $\mathcal{T}$.\;
2. \For{$s,t\in V$ such that $s\neq t$}{
$\widehat{d}(x,y) \leftarrow$ Recursive-APSD$(\mathcal{T}, G, (s,t), 0)$. \hfill \tcp{Algorithm \ref{alg:recursive-apsp}}
}
\Output{estimated distances $\widehat{d}(x,y)$ on $G$.}

\caption{Differentially private all-pair-shortest-path approximation}
\label{alg:apsp}
%\algolab{private_path}
\end{algorithm2e}

\subsection{Privacy Analysis}\label{sec:privacy}

Here, we analyze the privacy guarantee of our algorithm. We observe that constructing the binary tree $\mathcal{T}$ does not compromise privacy, as it only requires the topology, not the edge weights, as input. Consequently, the recursive procedure (Algorithm \ref{alg:recursive-apsp}) and Algorithm \ref{alg:apsp} involving the tree $\mathcal{T}$ are simply post-processing steps of Algorithm \ref{alg:shortcut}. Therefore, we focus our privacy analysis solely on Algorithm \ref{alg:shortcut}. We use the following lemma to bound the sensitivity of the binary tree that traces the decomposition: 

\begin{lemma}\label{lem:tree_sensitive}
Let $\mathcal{T}$ be the labeled tree constructed from $G$. For any vertex $u,v$ with an edge $e = \{u,v\}$ between them, $e$ appears in at most $h = \log_{1/q'}(n/c)$ nodes in $\mathcal{T}$.
\end{lemma}

\begin{proof}
Suppose $e$ is in the edge set of some internal node $G_b$. Recall that in the division of $G_b$, the edge $e$ will be removed if and only if $u,v\in S_b$. Then $e$ only exists in any $G_{b'}$ such that $b'$ is the prefix of $b$. Since $|b| \leq h$, the lemma holds.

On the other hand, if $e$ is in the edge set of some leaf node $G_b$. Then the recursive construction terminates on $G_b$, and $e$ only exists in any $G_{b'}$ such that $b'$ is the prefix of $b$, again the number of such nodes is at most $h$. This completes the proof of \Cref{lem:tree_sensitive}.
\end{proof}

We are now ready to prove the following theorem on the privacy guarantee, utilizing the advanced composition lemma (Lemma \ref{lem:adv_composition}).

\begin{theorem}\label{thm:privacy}
Algorithm \ref{alg:apsp} preserves $(\epsilon, \delta)$-differential privacy. 
\end{theorem}

\begin{proof}
    Suppose for a pair of neighboring graphs $G$ and $G'$, they have a difference in the edge weight by $1$ of $e = \{u,v\}$. Let $G_{b}$ be the subgraph of $G$ that contains $e$ in its edge set. 
    \begin{enumerate}
        \item If $G_b$ is an internal node of $\mathcal{T}$. let $d_{S_b} = \{d_b(x,y):x\neq y \land (x,y)\in S_b\}$ be the vector of distances between vertices in the separator of $G_b$. Since 
    $|d_b(x,y) - d_b'(x,y)|\leq 1$ for any distinct $x,y\in S_b$, then 
    $\|d_{S_b}\|_2 \leq p$ as $|S_b\times S_b| \leq p^2$. Then, by the Gaussian mechanism, outputting $\widehat{d}_b(x,y)$ for all distinct $x,y\in S_b$ preserves $(\epsilon',\delta')$-differential privacy. Similarly, outputting $\widehat{d}_b(x,y)$ for all $(x,y)\in S_b\times S_{b'}$ where $x,y\notin S_{b'}\cap S_b$ also preserves $(\epsilon',\delta')$-differential privacy, as $|S_b\times S_{b'}| \leq p^2$.
    \item If $G_b$ is a leaf node of $\mathcal{T}$. Since $V_b \leq c$, then again be the Gaussian mechanism, outputting the all-pair-shortest distance $\widehat{d}_b(\cdot, \cdot )$ in $G_b$ preserves $(\epsilon',\delta')$-differential privacy.
    \end{enumerate}

    Combining the above argument, Lemma \ref{lem:tree_sensitive} and the advanced composition, we have that Algorithm \ref{alg:shortcut} preserves $(\epsilon,\delta)$-DP, as well as Algorithm \ref{alg:apsp}. This completes the proof of \Cref{thm:privacy}.
\end{proof}

\subsection{Utility Analysis}\label{sec:utility}
Here, we present the utility guarantee of Algorithm \ref{alg:apsp}, which privately computes the all-pair shortest distances for the input graph by repeatedly invokes Algorithm \ref{alg:recursive-apsp} for each pair. We defer the proof in Theorem~\ref{thm:utility} in Appendix~\ref{app:utility}.

\begin{theorem}\label{thm:utility}
    Fix any $0<c\leq n$ and any $0<\gamma<1$. Let $G = (V,E,w)$ be a $(p,q,q')$-recursively separable graph for some $p\in \mathbb{N}$ and $\frac{1}{2}\leq q\leq q' <1$. Then with probability at least $1-\gamma$, we have that for any $s,t\in V$, 
    $$|\widehat{d}(s,t) - d(s,t)| \leq O\left(\frac{(hp + c)\cdot \log\left({h}/{\delta}\right)\cdot \sqrt{h^2+h\log (\max\{p,c\}) + h\log(1/\gamma)} }{\epsilon} \right),$$
    where $h = \log_{1/q'}(n/c)$. That is, for any constant $c$ and $q'$, we have with high probability, 
     $$|\widehat{d}(s,t) - d(s,t)| \leq O\left( \frac{p  \log^2 n\cdot \log\left(\frac{\log n}{\delta}\right)}{\epsilon} \right).$$
\end{theorem}
%\jalaj{Can you check this? We should  have a factor of $\log \log(n/\delta)$. More explanation: $p\leq n$, so $h\log(max(p,c)) = h \log n = O(h^2) = O(\log^2 n)$ for constant c. Also $hp=p \log(n)$ and $\log(h/\delta) = \log(\log n/\delta)$.}
%\zongrui{I see. Yes, there should be an extra log log factor.}

\paragraph{The corresponding result for pure-DP.}  Our framework can be easily adapted to pure differential privacy by (1) replacing Gaussian noise with Laplace noise in the shortcuts, and (2) employing basic composition instead of advanced composition. Using the Laplace mechanism and the tail bound for Laplace noise, one can directly derive the following theorem from our framework. The proof is therefore omitted.
\begin{theorem}
Fix a $p\in \mathbb{N}$. Let $G = (V,E,w)$ be a $(q,p,p')$-recursively separable graph for some constant $\frac{1}{2}\leq q\leq q' <1$. Then, there is a $(\epsilon,0)$-differentially private algorithm on estimating APSD such that with probability at least $1-\gamma$,
 $$|\widehat{d}(s,t) - d(s,t)| \leq O\left( \frac{p^2 \cdot \log^{2} n(\log n + \log (1/\gamma))}{\epsilon} \right).$$
\end{theorem}

\section{Improved Results for \texorpdfstring{$K_h$}{Lg}-Minor-Free Graphs}\label{sec:planar}

In this section, we utilize the algorithmic framework proposed in Section \ref{sec:recursive} to present our improved results on privately answering APSD for some special classes of graphs. In particular, assuming bounded weights, where the maximum weight of edges is denoted as \( W \), and by some specific separation theorems for special graphs, we can use the general framework (along with a sub-sampling trick inside the separator) to achieve a purely additive error of $\sim (nW)^{1/3}$ when privately approximating APSD in $K_h$-minor-free graphs for any constant $h$. Additionally, as a special subclass of $K_5$-minor free graphs, we are able to leverage some unique structural properties of grid graphs to further reduce the additive error to approximately $\sim n^{1/4}\sqrt{W}$.

\begin{theorem}\label{thm:kh_free_and_planar}
    Fix any \( 0 < \epsilon, \delta < 1 \) and \( n \in \mathbb{N} \). There exists an \((\epsilon, \delta)\)-differentially private algorithm for estimating all-pairs shortest path distances in \( G \) such that, with high probability, the worst-case error is bounded by 
\[
O\left( \frac{n^{1/4}\sqrt{W} \cdot \log^2 n }{\epsilon^{1/2}} \log\left(\frac{\log n}{\delta}\right) \right)
\]
for any grid graph \( G = ([a\times b], E, w) \) where $ab = n$. Further, the error is bounded by 
\[
O\left(h\cdot  \frac{(nW)^{1/3}\cdot \log^2 n }{\epsilon^{2/3}} \log\left(\frac{\log n}{\delta}\right) \right)
\]
with high probability, for any $K_h$-minor-free graph \( G = ([n], E, w) \).
\end{theorem}
Clearly Theorem \ref{thm:main} does not directly yield the results stated above. Therefore, some modifications to the recursive framework that is introduced earlier are necessary to derive Theorem \ref{thm:kh_free_and_planar}. Specifically, instead of constructing the all-pair shortcut inside or between separators as in Algorithm \ref{alg:shortcut}, we find a $k$-covering set of each separator and only add shortcuts inside or between such covering sets to reduce the number of compositions needed to preserve privacy. 

\begin{definition} [$k$-covering]\label{def:covering_set} Given a graph $G = (V,E)$, a subset $Z \subseteq V$ is a $k$-covering of $V$ if for every vertex $a \in V$, there is a vertex $b \in Z$ such that the hop distance between $a$ and $b$ is at most $k$.
\end{definition}

To prove Theorem \ref{thm:kh_free_and_planar}, we begin by presenting the decomposition procedure for $K_h$-minor-free graphs, which relies on the well-known lemma characterizing the separability of such graphs:

\begin{lemma}[Alon et al. \cite{alon1990separator}]\label{lem:minor_free_separator} Let $h\in \mathbb{N}_+$ be an integer, and $G$ be a $K_h$-minor-free graph on $n$ vertices. Then, there exists a separator $S$ of $G$ of order at most $O(h^{3/2}n^{1/2})$ such that no connected components in $G(V\backslash S, E(V\backslash S))$ has more than $\frac{2}{3} n$ vertices.
\end{lemma}

\noindent \textbf{The construction of the binary tree for $K_h$-minor-free graphs} is similar to that for general graphs, and that a $K_h$-minor-free graph \( G \) with \( n \) vertices are guaranteed to always have a (not necessarily connected) separator \( S \) of size at most \( O(\sqrt{h^3n}) \), such that removing \( S \) results in two disjoint subgraphs \( G'_0 \) and \( G'_1 \). Both subgraphs are clearly also $K_h$-minor-free, and are of size at most \( 2n/3 \). This process can be repeated recursively, splitting each subgraph until every part has constant size. For any grid graph on $n$ vertices of shape $a\times b$, we can split the graph into two sub-grid graphs of (almost) equal size using a separator of size at most \( O(\sqrt{n}) \), and that the separator is \emph{connected}. Hence we have the lemma below.

\begin{lemma}
    Fix an integer $h\geq 1$, any $K_h$-minor-free graph $G = ([n], E)$  is $(c\sqrt{h^3n}, 2/3, 2/3+ch^{3/2}/\sqrt{n})$-recursively separable for some constant $c$. In particular, planar graphs (as well as grid graphs $G_{\text{grid}} = ([a\times b], E)$ with $ab = n$) are $(c'\sqrt{n}, 2/3, 2/3+o(1))$-recursively separable for some constant $c'$.
\end{lemma}

As we introduced earlier, with this decomposition, the only modification required is in Algorithm \ref{alg:shortcut}. In Algorithm~\ref{alg:shortcut}, we construct private shortcuts by connecting all pairs of vertices in separator. We now modify Algorithm~\ref{alg:shortcut} into Algorithm~\ref{alg:shortcut_planar} as follows: for each separator, we {begin} by finding a \( k \)-covering for it. Then, we construct private shortcuts by connecting all pairs of vertices \emph{within} the $k$-covering sets. We provide the complete pseudocode for Algorithm \ref{alg:shortcut_planar} in Appendix \ref{app:alg4}, and here we only highlight the differences. The modified step 4 of Algorithm~\ref{alg:shortcut} is as follows.

\medskip
\fbox{
\parbox{0.9\textwidth}{
\noindent\textbf{Step 4 of Algorithm~\ref{alg:shortcut_planar} (modified step 4 of Algorithm~\ref{alg:shortcut})}:

\For{non-leaf node $(G_b, S_b) \in V(\mathcal{T})$}{

    Find a $k$-covering set of $S_b$, and let it be $S^k_b$.\;
    
    \For{$x,y \in S^k_b$ such that $x \neq y$}{
        \texttt{IsShortcut}$(x,y,b) = \textsf{\em True}$.\;
        Let $\widehat{d}_b(x,y) = d_b(x,y) + \mathcal{N}(0,\sigma^2)$.\;
    }
    \If{$b \neq \emptyset$}{
        Let $b'$ be the binary string that removes the last bit in $b$.\;
        
        Find a $k$-covering set of $S_{b'}$, and let it be $S^k_{b'}$.\;
        
        \For{$(x,y) \in S^k_{b'} \times S^k_b$}{
            \If{$x,y \notin S^k_{b'} \cap S^k_b$}{
                \texttt{IsShortcut}$(x,y,b') = \textsf{\em True}$.\;
                Let $\widehat{d}_b(x,y) = d_b(x,y) + \mathcal{N}(0,\sigma^2)$.\;
            }
        }
    }
}
}
}

\medskip
We also modify Step 3 of Algorithm~\ref{alg:shortcut} to revise the number of compositions needed for privacy from $O(p^2)$ to $O(f^2(p,k))$, as follows:

\medskip
\noindent\textbf{Step 3 of Algorithm~\ref{alg:shortcut_planar}}:
Let $\sigma = f(p,k)\sqrt{2\log(1.25/\delta')}/\epsilon'$, where $f(p,k)$ is the upper bound of the size of the $k$-covering set for the separator of size $p$. %\jalaj{Also this is changed for composition theorem, right? Might be worthwhile to mention it before stating the changed thing above and why we changed this step.} 

\subsection{Privacy Analysis}

Based on the proof of Theorem~\ref{thm:privacy}, we give the following privacy guarantee on Algorithm \ref{alg:shortcut_planar}.   

\begin{theorem}\label{thm:privacy_planar}
Algorithm~\ref{alg:shortcut_planar} preserves $(\epsilon, \delta)$-differential privacy.
\end{theorem}

The proof of Theorem~\ref{thm:privacy_planar} is almost identical to that of Theorem~\ref{thm:privacy}, except that we re-calibrate the variance of Gaussian noise according to the size of the $k$-covering set of a separator $S$, instead of the size of the whole separator.

\subsection{Utility Analysis}

To demonstrate the improvement in finding $k$-covering sets within separators, we first present several useful facts and lemmas related to $k$-covering. The following lemmas state that each connected $n$-vertex graph has a $k$-covering set of size at most $O(n/k)$.

\begin{lemma}[Meir and Moon~\cite{pjm1102868240}]\label{lem:cite} 
Any connected undirected graph with $n$ vertices has a $k$-covering with size at most $1+\lfloor n/(k+1) \rfloor$.
\end{lemma}

\begin{lemma}
Any graph with $n$ vertices composed of $x$ connected components has a $k$-covering with size at most $x + \lfloor n/(k+1) \rfloor$.
\end{lemma}
\begin{proof}

Let $x$ be the number of components with sizes $n_1, n_2, \dots, n_x$, respectively. Each component has a $k$-covering set of size $1+\lfloor n_i / (k+1) \rfloor)$ for $i = 1, 2, \dots, x$. Summing these sizes yields the desired result.
\end{proof}

We use the above lemmas to prove the following structural result, which establishes a covering set for the separator of $K_h$-minor-free graphs. We defer the proof of Lemma~\ref{lem:k-cover_in_separator} in Appendix~\ref{app:proof_of_lem25}.

\begin{lemma}
[Covering lemma for minor free graphs]
\label{lem:k-cover_in_separator}
    Fix an $h\geq 1$. Let $G = ([n],E)$ be a connected $K_h$-minor-free graph. Then for any $1\leq d \leq n$, there exists a subset of vertices $S\subseteq [n]$ such that: (1) $S$ is a separator of $G$ and (2) there is a $d$-covering of $S$ with size at most $O(\sqrt{h^3n/d})$.
\end{lemma}

\begin{remark}
We note that from Kuratowski's theorem \cite{kuratowski1930probleme}, each planar graph is $K_5$-minor-free. Therefore, Lemma \ref{lem:k-cover_in_separator} directly implies that each connected $n$-vertex planar graph has a separator with a $d$-covering of size $O(\sqrt{n/d})$ for any $1\leq d \leq n$.
\end{remark}

\begin{lemma} \label{lem:hop_error}
For any two vertices $a$ and $b$ in a graph $X$, and $z_a, z_b \in Z$, where $Z$ is a $k$-covering of $X$, such that the hop distance between $a$ (resp. $b$) and $z_a$ ($z_b$) is at most $k$, we have:
\[
|d(a,b) - d(z_a, z_b)| \leq 2k \cdot W
\]
where $W$ is the maximum weight of edges in the graph.
\end{lemma}
\begin{proof}
The lemma follows using the following set of inequalities:
$|d(a,b) - d(z_a, z_b)| \leq |d(z_a,z_b)+d(a,z_a)+d(b,z_b) - d(z_a, z_b)| \leq 2k \cdot W.$
\end{proof}

%To bound the size of the $k$-covering for the separator, we introduce the following lemmas.

%\begin{lemma}\label{lem:planar}
%For any planar graph with $n$ vertices, one can find a separator $S$ with size $|S| = O(n/l)$, and $S$ has $O(l)$ components~\cite{Grigni2004WellConnected}.
%\end{lemma}

%\begin{lemma}
%For a grid graph, the function $f(p,k) = O\left(\frac{p}{k}\right)$, whereas for a planar graph, the function $f(p,k) = O\left(\frac{p}{k} + \frac{n}{p}\right)$, where $n$ is the size of the graph, and $p$ is the size of the separator.
%\end{lemma}
%\begin{proof}

%For grid graphs,
%\[
%f(p,k) = O\left(\frac{p}{k}\right).
%\]
%This result is derived from the fact that the separator for the grid is connected, as shown in Lemma~\ref{lem:cite}. For a planar graph, a separator of size \( p \) partitions the graph into \( O(n/p) \) connected components. Thus, for planar graphs, the function becomes:
%\[
%f(p,k) = O\left(\frac{p}{k} + \frac{n}{p}\right),
%\]
%based on Lemma~\ref{lem:planar}.

%\end{proof}

The following lemmas are used to derive the error bound of the recursion based on two types of separators. The proofs of Lemma~\ref{lem:resursive_equation_new}, Lemma~\ref{lem:cross-level-shortcut_new}, and Lemma~\ref{Lem:error_x} are deferred to Appendix~\ref{app:proof_of_lem28}, Appendix~\ref{app:proof_of_lemma29}, and Appendix~\ref{app:proof_of_lemma30}, respectively.

\begin{lemma} \label{lem:resursive_equation_new}
Fix any $s,t \in V_b$, we have:
\[
d_b(s,t) = \min\{d_{b \circ 0}(s,t), d_{b \circ 1}(s,t)\}
\]
or
$$d_b(s,t) \leq \min_{x,y \in S^{k}_b} (d_{b \circ 0}(s,x) + d_b(x,y) + d_{b \circ 1}(y,t))  + 2kW.$$
\end{lemma}

\begin{lemma} \label{lem:cross-level-shortcut_new}
Let $S_{b'}$ and $S_b$ be two adjacent separators, where $b = b' \circ u$ for $u \in \{0,1\}$. Then for any $x \in S_b$, $a \in \{0,1\}^n$, and $y \in S_{b'}$, we have:
%\jalaj{Recall the reader that $S_b^k$ is as defined in step 4 of the algorithm 4. Or high level what it is. Makes life of reader easy}
\[
d_b(x,y) \leq \min_{z \in S^k_b} d_b(x,z) + d_{b \circ a}(z,y) \leq d_b(x,y) + 2kW.
\]
Recall that $S_b^k$ is as defined in Step 4 of Algorithm~\ref{alg:shortcut_planar}.
\end{lemma}
\noindent Next, we write $\sigma'_1 = c\cdot \sqrt{2\log(1.25/\delta')}/\epsilon'$ and $\sigma'_2 = f(p,k)\cdot \sqrt{2\log(1.25/\delta')}/\epsilon'$.
\begin{lemma}\label{Lem:error_x}
Let $h,p,k,c,\gamma$ be as before and $g(p,k,c,\gamma,h):=\sqrt{{h + 3\ln (\max\{f(p,k),c\}) + \ln\left(\frac{1}{2\gamma}\right)}}$. 
With probability at least $1 - \gamma$, the following holds for Algorithm~\ref{alg:shortcut_planar}:

    \begin{enumerate}
        \item For any $s,t,b$ where $|b| = h$ and that $\texttt{IsShortcut}(s,t,b) = \textsf{\em True}$,
    $$|d_b(s,t) - \widehat{d}_b(s,t)| \leq \sqrt{2}\sigma'_1\cdot g(p,k,c,\gamma,h).$$
    \item For any $s,t,b$ where $|b| < h$ and that $\texttt{IsShortcut}(s,t,b) = \textsf{\em True}$,   $$|d_b(s,t) - \widehat{d}_b(s,t)| \leq \sqrt{2}\sigma'_2 \cdot g(p,k,c,\gamma,h) 
    %\sqrt{{h + 3\ln (\max\{f(p,k),c\}) + \ln\left(\frac{1}{2\gamma}\right)}}
    .$$
    \end{enumerate}
\end{lemma}

To prove the utility guarantee for Algorithm~\ref{alg:shortcut_planar}, we now state the following error bound which follows directly from Lemma~\ref{lem:resursive_equation_new}, Lemma~\ref{lem:cross-level-shortcut_new} and Lemma~\ref{Lem:error_x}.

\begin{lemma}\label{lem:error_cover}
     Fix any $0<\epsilon,\delta<1$, and $n,p\in \mathbb{N}$. For any planar graph $G = ([n],E,w)$ that is $(p,q,q')$-recursively separable for some constants $\frac{1}{2} \leq q \leq q' <1$, there is an $(\epsilon,\delta)$-algorithm for estimating all-pair shortest distances in $G$ such that with high probability,
    $$|\widehat{d}(s,t) - d(s,t)| \leq O\left( \frac{f(p,k) \cdot \log^2 n\log(n/\delta)}{\epsilon} +kW \right).$$
\end{lemma}

By applying Lemma \ref{lem:k-cover_in_separator} to bound $f(p,k)$ in Lemma \ref{lem:error_cover}, we are now ready to present the proof of the error bound stated in Theorem \ref{thm:kh_free_and_planar}.

\begin{proof} 
[Proof Of Theorem \ref{thm:kh_free_and_planar}]
The privacy guarantee follows from Theorem~\ref{thm:privacy_planar}. Now we give the proof for the utility guarantee.

(1) For a grid graph, let \( p = \sqrt{n} \), the decomposition described in Section \ref{sec:planar} gives that \( f(p, k) = O( n^{1/4}\sqrt{W\epsilon}) \) when \( k = n^{1/4}/\sqrt{W\epsilon} \). Then, the result can be obtained by substituting this into Lemma~\ref{lem:error_cover}. This choice of parameters is optimal due to the arithmetic-geometric mean inequality.

(2) For a $K_h$-minor-free graph, we choose \( k = O(\frac{n^{1/3}}{(\epsilon W)^{2/3}}) \) to balance the number of hops required to reach the covering set for each vertex in the separator and the size of the \( k \)-covering of the separator. In this case, from Lemma~\ref{lem:k-cover_in_separator}, we have \( f(p, k) = O(h\cdot \left( nW\epsilon \right)^{1/3} )\). Then, the desired result follows by substituting this into Lemma~\ref{lem:error_cover}.

\end{proof}

%\section{Private APSD in Graphs Without Large $K_5$ or $K_{3,3}$ Subdivisions}

%\section{Conclusion}

%Grid graph is special that it has perfect separation with  connected separator.

%\noindent\textbf{The construction of binary tree for grid graphs} is similar as the construction for general graphs. The key difference is that: Grid graph $G$ with size $n$ (total number of vertices)  always has a connected separator $S_b$ with size at most $\sqrt{n}$ such that removing resulting in 
%two disjoint  subgraphs $G'_0$ and $G'_1$ also be grid graphs, and the size of each subgraphs is at most $n/2$. In this way we continue to split  each subgraphs recursively until each part has  constant size.

%In Algorithm~\ref{alg:shortcut}, we construct private shortcuts by connecting all pairwise vertices between vertices in $S_b\cup S'_b$. 

\clearpage
\bibliographystyle{alpha}
\bibliography{privacy}
\clearpage
\appendix

\section{Proof of \texorpdfstring{Theorem~\ref{thm:utility}}{Lg}}\label{app:utility}

In this section, we prove Theorem~\ref{thm:utility}. First, we observe that, with probability 1, Algorithm \ref{alg:recursive-apsp} terminates normally without any abnormal termination.

\begin{fact}\label{fac:no_fail}
    For any pair of vertices $s,t\in V$, Algorithm \ref{alg:recursive-apsp} does not output ``{\em \texttt{FAIL}}''.
\end{fact}
\begin{proof}
    The only case when Algorithm \ref{alg:recursive-apsp} outputs ``\texttt{FAIL}'' is that the parameter $k\geq 1$ and both $s,t$ are not in the separator of the predecessor graph of $G_b$. However, it is not possible because when Algorithm \ref{alg:recursive-apsp} is called with $k\geq 1$, at least one of the vertex in $s,t$ is from the separator $S_{b'}$ where $G_b = G_{b'\circ 0}$ or $G_b = G_{b'\circ 1}$.
\end{proof}

\noindent The following lemma is used for building the correctness of the recursion.
\begin{lemma}\label{lem:resursive_equation}
    Fix any $s,t\in V_b$. Without the lose of generality, we assume either both $s,t\in V_{b\circ 0}$ or $s\in V_{b\circ 0}$ and $t\in V_{b\circ 1}$ (otherwise we just switch $s$ and $t$, and the proof for both $s,t\in V_{b\circ 1}$ is symmetric). We have
    $$d_b(s,t) = \min\left\{d_{b\circ 0}(s,t), d_{b\circ 1}(s,t), \min_{x,y\in S_b} d_{b\circ 0}(s,x) + d_{b}(x,y) + d_{b\circ 1}(y,t)\right\}.$$
\end{lemma}
\begin{proof}
    We discuss by difference cases:
    
    \noindent \textbf{1. When at least one of $s,t$ is in $S_b$}. By the construction of $G_{b\circ a}$ such that $$G_{b\circ a} = (V_{b\circ a}' \cup S, E(V_{b\circ a}\cup S_b) \backslash E(S_b)),$$
    we have $d_b(s,t) \geq d_{b\circ 0}(s,t)$ and $d_b(s,t) \geq d_{b\circ 1}(s,t)$. Similarly, we have that
    \begin{equation}\label{eq:lem9-1}
        \begin{aligned}
        \min_{x,y\in S_b} d_{b\circ 0}(s,x) + d_{b}(x,y) + d_{b\circ 1}(y,t) \geq\min_{x,y\in S_b} d_{b}(s,x) + d_{b}(x,y) + d_{b}(y,t) \geq d_b(s,t)
        \end{aligned}
    \end{equation}
    by the triangle inequality. On the other hand, if both $s,t\in S_b$, then we have

        \begin{equation}\label{eq:lem9-2}
        d_b(s,t) =  d_{b\circ 0}(s,s) + d_{b}(s,t) + d_{b\circ 1}(t,t) \geq \min_{x,y\in S_b} d_{b\circ 0}(s,x) + d_{b}(x,y) + d_{b\circ 1}(y,t).
    \end{equation}
    If exactly one of $s,t\in S_b$, we note that
    any path $(v_1 = s, v_2, \cdots, v_k = t)$ from $s$ to $t$, must satisfy that $v_i \in S_b$ for at least one $i\in[k]$ since $s\in S_b$ or $t\in S_b$. Again by the symmetry, we assume $t\in S_b$. Suppose $z$ is the first vertices in $S_b$ of any shortest path from $s$ to $t$. Then,
        \begin{equation}\label{eq:lem9-3}
        d_b(s,t) =  d_{b\circ 0}(s,z) + d_{b}(z,t) + d_{b\circ 1}(t,t) \geq \min_{x,y\in S_b} d_{b\circ 0}(s,x) + d_{b}(x,y) + d_{b\circ 1}(y,t).
    \end{equation}
    Combining \Cref{eq:lem9-1}, \Cref{eq:lem9-2} and \Cref{eq:lem9-3} completes the proof in this case.

    \noindent \textbf{2. When both $s,t \notin S_b$.} First, we assume that $s$ and $t$ are on different sides such that $s\in V_{b\circ 0}$ and $t\in V_{b\circ 1}$. In this case $d_{b\circ 0}(s,t) = d_{b\circ 1}(s,t) = \infty$. By the fact that $S_b$ is a separator, any path $(v_1 = s, v_2, \cdots, v_k = t)$ from $s$ to $t$ must also satisfy that $v_i \in S_b$ for at least one $i\in[k]$. Let $z_1, z_2$ be the first and last vertices in one of the shortest paths from $s$ to $t$ such that $z_1,z_2\in S_b$ (we allow $z_1 = z_2$). Then it is easy to verify that 
    $d_b(s,z_1) = d_{b\circ 0}(s,z_1)$ and  $d_b(z_2,t) = d_{b\circ 1}(z_2,t)$. Thus, 
    $$\min_{x,y\in S_b} d_{b\circ 0}(s,x) + d_{b}(x,y) + d_{b\circ 1}(y,t) \leq  d_{b\circ 0}(s,z_1) + d_{b}(z_1,z_2) + d_{b\circ 1}(z_2,t) = d_b(s,t).$$
    Again by the triangle inequality, we also have \Cref{eq:lem9-1} holds, which proves Lemma \ref{lem:resursive_equation}. Now, suppose that both $s$ and $t$ are in the same side $V_{b\circ 0}$. Then the shortest path either crosses $S_b$ or not. If the path crosses $S_b$, then the previous argument suffices to prove Lemma \ref{lem:resursive_equation}. If the path does not crosses $S_b$, then we have $d_b(s,t) = \min\{d_{b\circ 0}(s,t), d_{b\circ 1}(s,t)\}$, which completes the proof.
    
\end{proof}

\noindent The following lemma utilizes the shortcut between the separators of two adjacent layers (built by step 4 of Algorithm \ref{alg:shortcut}) for pruning.

\begin{lemma}\label{lem:cross-level-shortcut}
Let $S_{b'}$ and $S_{b}$ be two adjacent separators where $b = b'\circ u$ for $u\in \{0,1\}$. Then for any $x\in S_{b}$, $a\in \{0,1\}^n$ and $y\in S_{b'}$,
$$d_b(x,y) = \min_{z\in S_{b}} d_{b}(x,z) + d_{b\circ a}(z,y).$$
\end{lemma}
\begin{proof}
    Let $P = (v_1 = x, v_2, \cdots, v_k = y)$ be any shortest path from $x$ to $y$, and let $w$ be the last vertex in $P$ such that $w\in S_{b}$. Such $w$ exists since $x\in S_{b}$. Then, $$d_{b}(x,y) = d_{b}(x,w) + d_{b}(w,y) = d_{b}(x,w) + d_{b\circ a}(w,y) \geq  \min_{z\in S_{b}} d_{b}(x,z) + d_{b\circ a}(z,y).$$ 
    On the other hand, 
    $$d_b(x,y) \leq \min_{z\in S_{b}} d_{b}(x,z) + d_{b}(z,y) \leq \min_{z\in S_{b}} d_{b}(x,z) + d_{b\circ a}(z,y)$$
    due to the triangle inequality and the fact that $d_{b}(z,y) \leq d_{b\circ a}(z,y)$. This completes the proof of Lemma \ref{lem:cross-level-shortcut}.
\end{proof}
\noindent Without the lose of generality, we assume the $\mathcal{T}$ is a complete binary tree of height $h$. (If one branch terminates early, we just let it continue to split with one of its branches be an empty graph, until the height is $h$.) We need to use the following observation to control the accumulation of the error:
\begin{fact}\label{fac:soundness_of_alg2}
    During the execution of Algorithm \ref{alg:apsp}, for any $b\in \{0,1\}^*$ such that $|b| \leq h$ and any pair of vertices $s,t\in V_b$, Algorithm \ref{alg:recursive-apsp} with parameter $G_b$ and $(s,t)$ is invoked exactly once.
\end{fact}
\begin{proof}
    We prove it by induction on the size of $b$. We claim that if Recursive-APSD$(\mathcal{T},G_{b}, (s,t), k)$ where $|b|\leq h-1$ is called for any $s,t\in G_b$ and some $k$, then both Recursive-APSD$(\mathcal{T},G_{b\circ 0}, (s,t), k')$ and Recursive-APSD$(\mathcal{T},G_{b \circ 1}, (s,t), k')$ will be called for some $k'$. This is clearly true since Recursive-APSD$(\mathcal{T},G_{b\circ a}, (s,t), k+1)$ will be invoked if both $s$ and $t$ lies in $G_{b\circ a}$. Then, since Algorithm \ref{alg:apsp} invokes Algorithm \ref{alg:recursive-apsp} for any distinct $s,t\in V$, we conclude that Algorithm \ref{alg:recursive-apsp} will be invoked at least once for every $G_b$ and all-pair vertices in $G_b$. On the other hand, since Algorithm \ref{alg:recursive-apsp} with $G_b$ as parameter will only be invoked by $b'$ such that $b = b'\circ a$ for $a\in \{0,1\}$, then it will only be invoked once. %In Algorithm \ref{alg:apsp}.
\end{proof}

Next, we analyze the error on each shortcut. We write $\sigma_1 = c\cdot \sqrt{2\log(1.25/\delta')}/\epsilon'$ and $\sigma_2 = p\cdot \sqrt{2\log(1.25/\delta')}/\epsilon'$.

\begin{lemma}\label{lem:error_in_shortcut}
    With probability at least $1-\gamma$, both the following holds:
    \begin{enumerate}
        \item For any $s,t,b$ where $|b| = h$ and that $\texttt{IsShortcut}(s,t,b) = \textsf{\em True}$,
    $$|d_b(s,t) - \widehat{d}_b(s,t)| \leq \sigma_1\cdot \sqrt{2\Paren{h + 3\ln (\max\{p,c\}) + \ln\left(\frac{1}{2\gamma}\right)}}$$
    \item For any $s,t,b$ where $|b| < h$ and that $\texttt{IsShortcut}(s,t,b) = \textsf{\em True}$,   $$|d_b(s,t) - \widehat{d}_b(s,t)| \leq \sigma_2 \cdot \sqrt{2\Paren{h + 3\ln (\max\{p,c\}) + \ln\left(\frac{1}{2\gamma}\right)}}.$$
    \end{enumerate}
\end{lemma}
\begin{proof}
    We recall that for the Gaussian variable $X\sim \mathcal{N}(0,\sigma^2)$, $\Pr{|X|\geq t}\leq \exp(-t^2/{2\sigma^2})$ for any $t\geq 0$. For any $s,t,b$ with $\texttt{IsShortcut}(s,t,b) = \text{ True}$ and $G_b$ is a leaf node, the difference in $\widehat{d}_b(s,t)$ and ${d}_b(s,t)$ is a Gaussian noise with variance $\sigma_1$, and thus
    $$\Pr{d_b(s,t) - \widehat{d}_b(s,t) \geq z} \leq \frac{\gamma}{m}$$
    if we choose $z = \sigma_1 \sqrt{2\log (m/(2\gamma))}$ for any $m\geq 1$ and $0<\gamma<1$. Also, if $G_b$ is not a leaf node, then according to Algorithm \ref{alg:shortcut}, we have 
     $$\Pr{d_b(s,t) - \widehat{d}_b(s,t) \geq \sigma_2 \sqrt{2\log m/(2\gamma)}} \leq \frac{\gamma}{m}.$$
     The number of noises added in Algorithm \ref{alg:shortcut} is bounded by $m = 5\cdot 2^h \cdot \max\{p^2,c^2\}$. Lemma \ref{lem:error_in_shortcut} now follows using the union bound.
    \end{proof}

With all the aforementioned preparations, we are now ready to prove Theorem~\ref{thm:utility} regarding the utility guarantee of Algorithm \ref{alg:apsp} by induction. We restate this theorem here.

\begin{theorem}
[Restatement of Theorem~\ref{thm:utility}]
    \label{thm:utility_restated}
     Fix any $0<c\leq n$ and any $0<\gamma<1$. Let $G = (V,E,w)$ be a $(p,q,q')$-recursively separable graph for some $p\in \mathbb{N}$ and $\frac{1}{2}\leq q\leq q' <1$. Then with probability at least $1-\gamma$, we have that for any $s,t\in V$, 
    $$|\widehat{d}(s,t) - d(s,t)| \leq O\left(\frac{(hp + c)\cdot \log\left({h}/{\delta}\right)\cdot \sqrt{h^2+h\log (\max\{p,c\}) + h\log(1/\gamma)} }{\epsilon} \right),$$
    where $h = \log_{1/q'}(n/c)$. That is, for any constant $c$ and $p'$, we have with high probability, 
     $$|\widehat{d}(s,t) - d(s,t)| \leq O\left( \frac{p \cdot \log^2 n\log(n/\delta)}{\epsilon} \right).$$
\end{theorem}
\begin{proof}
    With the help of Fact \ref{fac:soundness_of_alg2}, we prove this theorem by induction on the size of $b$. We first define 
    $$\texttt{err}(b) = \zeta_1 + (h - |b|)\zeta_2.$$
    for any $b\in \{0,1\}^*$ and $|b|\leq h$. Here, $\zeta_1 = \sigma_1\cdot \sqrt{2\paren{h + 3\ln (\max\{p,c\}) + \ln (1/(2\gamma))}}$ and $\zeta_2 = \sigma_2\cdot \sqrt{2\paren{h + 3\ln (\max\{p,c\}) + \ln (1/(2\gamma))}}$. Then, it is sufficient to just show that for any $G_b$ in the binary tree $\mathcal{T}$ and any $x,y\in V_b$,
    $$|d_b(s,t) - \widehat{d}_b(s,t)| \leq 2\texttt{err}(b),$$
    as letting $b = \emptyset$ completes the proof.
    For the base case, this is true for all $|b| = h$ because in this case $G_b$ is the leaf node of $\mathcal{T}$ and thus $|d_b(s,t) - \widehat{d}_b(s,t)| \leq \zeta_1 = \texttt{err}(b)$. For all $|b|<h$ and $\texttt{IsShortcut}(s,t,b) = \textsf{\em True}$, we also have that
    $$|d_b(s,t) - \widehat{d}_b(s,t)| \leq \zeta_2 \leq \zeta_1 + \zeta_2 \leq \texttt{err}(b).$$
    Suppose that for any $s,t$ and $b$ with $|b|<h$, we always have 
    \begin{enumerate}
        \item If $\widehat{d}_b(s,t)$ is computed by the unique invocation of Algorithm \ref{alg:recursive-apsp} with parameter $k > 0$, then  $|d_{b\circ a}(s,t) - \widehat{d}_{b\circ a}(s,t)| \leq \texttt{err}(b\circ a)$ for $a\in \{0,1\}$;
        \item If $\widehat{d}_b(s,t)$ is computed by the unique invocation of Algorithm \ref{alg:recursive-apsp} with parameter $k = 0$, then  $|d_{b\circ a}(s,t) - \widehat{d}_{b\circ a}(s,t)| \leq 2\texttt{err}(b\circ a)$ for $a\in \{0,1\}$.
    \end{enumerate}
   Now we look at the induction case.

    \noindent\textbf{Case(1)} Suppose $\widehat{d}_b(s,t)$ is computed by the unique invocation of Algorithm \ref{alg:recursive-apsp} with parameter $k > 0$ and that $\texttt{IsShortcut}(s,t,b) = \text{False}$. We only analysis the case when both $s,t\in V_{b\circ a}$ for $a\in \{0,1\}$, since the proof for the case where $s,t$ are in the different sides is identical. In this case, by Algorithm \ref{alg:recursive-apsp},  
    $$\widehat{d}_b(s,t) = \mathop{\min}\{\widehat{d}_{b\circ a}(s,t), \min_{x\in S_b} \widehat{d}_{b\circ a}(s,x) + \widehat{d}_b(x,t)\}.$$
    On the other hand, from Lemma \ref{lem:resursive_equation}, we have that
    $$d_b(s,t) = \min\left\{d_{b\circ a}(s,t), \min_{x,y\in S_b} d_{b\circ a}(s,x) + d_{b}(x,y) + d_{b\circ a}(y,t)\right\},$$
    where
    \begin{equation}
        \begin{aligned}
            \min_{x,y\in S_b} d_{b\circ a}(s,x) + d_{b}(x,y) + d_{b\circ a}(y,t) &= \min_{x\in S_b}\left( d_{b\circ a}(s,x) + \min_{y\in S_b}d_{b}(x,y) + d_{b\circ a}(y,t) \right)\\
        & = \min_{x\in S_b} d_{b\circ a}(s,x) + d_b(x,t).
        \end{aligned}
    \end{equation}
    Here, the second equality comes from Lemma \ref{lem:cross-level-shortcut} together with the fact that $t\in S_{b'}$ where $b'$ is the predecessor of $b$. Therefore, 
    $${d}_b(s,t) = \mathop{\min}\{{d}_{b\circ a}(s,t), \min_{x\in S_b} {d}_{b\circ a}(s,x) + {d}_b(x,t)\}.$$
    By the induction assumption, we 
    have both 
    $$|\widehat{d}_{b\circ a}(s,t) - {d}_{b\circ a}(s,t)| \leq \texttt{err}(b\circ a) \text{ and } |\widehat{d}_{b\circ a}(s,x) - {d}_{b\circ a}(s,x)| \leq \texttt{err}(b\circ a).$$
     Again by the fact that $t\in S_{b'}$ and $x\in S_b$, then $\texttt{IsShortcut}(x,t,b') = \textsf{\em True}$ and therefore $|\widehat{d}_b(x,t) - d_b(x,t)| \leq \zeta_2$. Combining these together we have
     $$|d_b(s,t) - \widehat{d}_b(s,t)|\leq \texttt{err}(b\circ a) + \zeta_2 \leq \zeta_1 + (h-(|b|+1))\zeta_2 +\zeta_2 = \texttt{err}(b).$$
    
    \noindent\textbf{Case(2)} Suppose $\widehat{d}_b(s,t)$ is computed by the unique invocation of Algorithm \ref{alg:recursive-apsp} with parameter $k = 0$ and that $\texttt{IsShortcut}(s,t,b) = \text{False}$. Still, we assume $s,t\in V_{b\circ a}$ for $a\in \{0,1\}$. Since for any $z\in S_b$, both Recursive-APSD$(\mathcal{T}, G_{b\circ a}, (s,z), k')$ Recursive-APSD$(\mathcal{T}, G_{b\circ a}, (t,z), k')$ will be invoked with $k'>0$, then from case (2), we have that for any $x,y\in S_b$, both 
    $\widehat{d}_{b\circ a}(s,x)$ and $\widehat{d}_{b\circ a}(y,t)$ is $\texttt{err}(b\circ a)$ far from ${d}_{b\circ a}(s,x)$ and ${d}_{b\circ a}(y,t)$ respectively. Also, for any $x,y\in S_b$, since $\texttt{IsShortcut}(x,y,b) = \textsf{\em True}$, then
    $$|\widehat{d}_b(x,y) - d_b(x,y)|\leq \zeta_2.$$
    From the induction assumption, we also have
    $$|\widehat{d}_{b\circ a}(s,t) - {d}_{b\circ a}(s,t) | \leq 2\texttt{err}(b\circ a).$$
    Combining these facts together Lemma \ref{lem:resursive_equation}, we have that
    \begin{equation}
        \begin{aligned}
            |\widehat{d}_b(s,t) - d_b(s,t)| 
            &= |\mathop{\min}\{\widehat{d}_{b\circ a}(s,t), \min_{x,y\in S_b} \widehat{d}_{b\circ a}(s,x) + \widehat{d}_b(x,y) + \widehat{d}_{b\circ a}(y,t)\}\\
            &- \mathop{\min}\{{d}_{b\circ a}(s,t), \min_{x,y\in S_b} {d}_{b\circ a}(s,x) + {d}_b(x,y) + {d}_{b\circ a}(y,t)\}|\\
            &\leq 2\texttt{err}(b\circ a) + \zeta_2 = 2\zeta_1 + 2(h-|b|)\zeta_2 - 2\zeta_2 + \zeta_2 \leq 2\texttt{err}(b).
        \end{aligned}
    \end{equation}
This finishes the proof of \Cref{thm:utility_restated}.    
\end{proof}

\section{Missing Proofs in Section~\ref{sec:planar}}

\subsection{Proof of Lemma~\ref{lem:k-cover_in_separator}}\label{app:proof_of_lem25}

\begin{lemma}
    [Restatement of Lemma~\ref{lem:k-cover_in_separator}]
    \label{lem:k-cover_in_separator_restate}
    Fix an $h\geq 1$. Let $G = ([n],E)$ be a connected $K_h$-minor-free graph. Then for any $1\leq d \leq n$, there exists a subset of vertices $S\subseteq [n]$ such that: (1) $S$ is a separator of $G$ and (2) there is a $d$-covering of $S$ with size at most $O(\sqrt{h^3n/d})$.
\end{lemma}
\begin{proof}
We first claim that for any connected graph on $n$ vertices, there exists a partition of $[n]$ into $s = O(n/d)$ disjoint subsets $(V_1, V_2, \cdots, V_{s})$ such that the diameter of each $V_i (1\leq i \leq s)$ is at most $d$. Indeed, by Lemma \ref{lem:cite}, $G$ has a $d$-covering $\mathcal{C}\subseteq [n]$ of size $O(n/d)$. For each $u\in \mathcal{C}$, let $V'_u \subseteq [n]$ be the vertices covered by $u$ with at most $d$ hops. Then, the partition $(V_1, V_2, \cdots, V_{s})$ can be constructed from $\{V'_u\}_{u\in \mathcal{C}}$ by removing duplicate elements. Next, by contracting each $V_i$ for $1 \leq i \leq s$ into a super-node and merging all multi-edges between the super-nodes, we obtain a smaller graph $\widetilde{G}$ with $s$ nodes, and clearly $\widetilde{G}$ is $K_h$-minor-free if $G$ is $K_h$-minor-free.

Thus, by the separation theorem for $K_h$-minor-free graphs (Lemma \ref{lem:minor_free_separator}), $\widetilde{G}$ has a separator of size $O(h^{3/2}\sqrt{s}) = O(\sqrt{h^3n/d})$. Picking one vertex from each super-node in the separator of $\widetilde{G}$ forms a $d$-covering set of the corresponding separator of the original graph, since each super-node has diameter at most $d$. This completes the proof of Lemma \ref{lem:k-cover_in_separator}.
\end{proof}

\subsection{Proof of Lemma~\ref{lem:resursive_equation_new}}\label{app:proof_of_lem28}

\begin{lemma} [Restatement of Lemma~\ref{lem:resursive_equation_new}]
\label{lem:resursive_equation_new_restate}
Fix any $s,t \in V_b$, we have:
\[
d_b(s,t) = \min\{d_{b \circ 0}(s,t), d_{b \circ 1}(s,t)\}
\]
or
\[
d_b(s,t) \leq \min_{x,y \in S^{k}_b} (d_{b \circ 0}(s,x) + d_b(x,y) + d_{b \circ 1}(y,t))  + 2kW.
\]
\end{lemma}

\begin{proof}
The proof is based on Lemma~\ref{lem:hop_error} and the proof of Lemma~\ref{lem:resursive_equation}. Without loss of generality, assume that either both $s,t \in V_{b \circ 0}$, or $s \in V_{b \circ 0}$ and $t \in V_{b \circ 1}$ (otherwise, we switch $s$ and $t$, and the proof for $s,t \in V_{b \circ 1}$ is symmetric)

   \noindent \textbf{1. When both $s,t \notin S^k_b$.} The proof is the same as the proof of the proof of Lemma~\ref{lem:resursive_equation}.

    \noindent \textbf{2. When at least one of $s,t$ is in $S^k_b$}. By the construction of $G^k_{b\circ a}$ such that $$G^k_{b\circ a} = (V_{b\circ a}' \cup S^k_b, E(V_{b\circ a}\cup S^k_b) \backslash E(S^k_b)),$$
    we have 
   $$  \min_{x,y \in S^{k}_b} (d_{b \circ 0}(s,x) + d_b(x,y) + d_{b \circ 1}(y,t)) \leq   \min_{x,y\in S_b} d_{b\circ 0}(s,x) + d_{b}(x,y) + d_{b\circ 1}(y,t) +2kW
   $$
Combining \Cref{eq:lem9-1}, \Cref{eq:lem9-2} and \Cref{eq:lem9-3} completes the proof in this case.
\end{proof}

\subsection{Proof of Lemma~\ref{lem:cross-level-shortcut_new}}\label{app:proof_of_lemma29}

\begin{lemma}[Restatement of Lemma~\ref{lem:cross-level-shortcut_new}] \label{lem:cross-level-shortcut_new_restate}
  Let $S_{b'}$ and $S_b$ be two adjacent separators, where $b = b' \circ u$ for $u \in \{0,1\}$. Then for any $x \in S_b$, $a \in \{0,1\}^n$, and $y \in S_{b'}$, we have:
\[
d_b(x,y) \leq \min_{z \in S^k_b} d_b(x,z) + d_{b \circ a}(z,y) \leq d_b(x,y) + 2kW.
\]
\end{lemma}

\begin{proof}
    Let $P = (v_1 = x, v_2, \cdots, v_k = y)$ be any shortest path from $x$ to $y$, and let $w$ be the last vertex in $P$ such that $w\in S_{b}$. Such $w$ exists since $x\in S_{b}$. Then, $$d_{b}(x,y) = d_{b}(x,w) + d_{b}(w,y) = d_{b}(x,w) + d_{b\circ a}(w,y) \geq  \min_{z\in S_{b}} d_{b}(x,z) + d_{b\circ a}(z,y) $$ $$ \geq \min_{z\in S^k_{b}} d_{b}(x,z) + d_{b\circ a}(z,y)-2kW.$$ 
    On the other hand, 
    $$d_b(x,y) \leq \min_{z\in S^k_{b}} d_{b}(x,z) + d_{b}(z,y) \leq \min_{z\in S^k_{b}} d_{b}(x,z) + d_{b\circ a}(z,y)$$
    due to the triangle inequality and the fact that $d_{b}(z,y) \leq d_{b\circ a}(z,y)$. This completes the proof.
\end{proof}

\subsection{Proof of Lemma~\ref{Lem:error_x}}\label{app:proof_of_lemma30}

\begin{lemma}[Restatement of Lemma~\ref{Lem:error_x}]\label{Lem:error_x_restate}
 Let $h,p,k,c,\gamma$ be as before and $$g(p,k,c,\gamma,h):=\sqrt{2\Paren{h + 3\ln (\max\{f(p,k),c\}) + \ln\left(\frac{1}{2\gamma}\right)}}.$$ 
With probability at least $1 - \gamma$, the following holds for Algorithm~\ref{alg:shortcut_planar}:

    \begin{enumerate}
        \item For any $s,t,b$ where $|b| = h$ and that $\texttt{IsShortcut}(s,t,b) = \textsf{\em True}$,
    $$|d_b(s,t) - \widehat{d}_b(s,t)| \leq \sigma'_1\cdot g(p,k,c,\gamma,h).$$
    \item For any $s,t,b$ where $|b| < h$ and that $\texttt{IsShortcut}(s,t,b) = \textsf{\em True}$,   $$|d_b(s,t) - \widehat{d}_b(s,t)| \leq \sigma'_2 \cdot g(p,k,c,\gamma,h) 
    %\sqrt{{h + 3\ln (\max\{f(p,k),c\}) + \ln\left(\frac{1}{2\gamma}\right)}}
    .$$
    \end{enumerate}
\end{lemma}
\begin{proof}
The proof is analogous to that of Lemma~\ref{lem:error_in_shortcut} in Appendix \ref{app:utility}. In particular, for any $s,t,b$ with $\texttt{IsShortcut}(s,t,b) = \text{ True}$ and $G_b$ is a leaf node, the difference in $\widehat{d}_b(s,t)$ and ${d}_b(s,t)$ is a Gaussian noise with variance $\sigma_1'$, and thus
    $$\Pr{d_b(s,t) - \widehat{d}_b(s,t) \geq z} \leq \frac{\gamma}{m}$$
    if we choose $z = \sqrt{2}\sigma_1' \sqrt{\log (m/(2\gamma))}$ for any $m\geq 1$ and $0<\gamma<1$. Also, if $G_b$ is not a leaf node, then according to Algorithm \ref{alg:shortcut_planar}, we have 
     $$\Pr{d_b(s,t) - \widehat{d}_b(s,t) \geq \sigma_2' \sqrt{2\log m/(2\gamma)}} \leq \frac{\gamma}{m}.$$
     The number of noises added in Algorithm \ref{alg:shortcut_planar} is bounded by $m = 5\cdot 2^h \cdot \max\{(f(p,k))^2,c^2\}$. Then, Lemma \ref{Lem:error_x} can be proved by the union bound.

%\jalaj{Rather writing this, keep the proof of this lemma in a place where you make this statement.}
 \end{proof}

\clearpage
\section{The Complete Pseudocode of Algorithm \ref{alg:shortcut_planar}}\label{app:alg4}

\begin{algorithm2e}[!ht]
\SetKwInput{Input}{Input}
\SetKwInOut{Output}{Output}
\Input{Graph $G=({V},{E},w)$, private parameter $\epsilon,\delta$, sampling parameter $k$.}
\DontPrintSemicolon
1. Recursively construct a binary tree $\mathcal{T}$ as described in this section.\;

2. Set $h = \log_{1/q'}(n/c)$,  $\epsilon' = \epsilon/\sqrt{4h\log(1/\delta')}$, and $\delta' = \delta/(4h)$.\;

3. Let $\sigma = f(p,k)\sqrt{2\log(1.25/\delta')}/\epsilon'$.\;

4. \For{non-leaf node $(G_b, S_b) \in V(\mathcal{T})$}{
Find a $k$-covering set of $S_b$ and let it be $S^k_b$.\;
\For{$x,y\in S^k_b$ such that $x\neq y$}{
\texttt{IsShortcut}$(x,y,b) = \textsf{\em True}$.\;
 
Let $\widehat{d}_b(x,y) = d_b(x,y) + \mathcal{N}(0,\sigma^2)$.\;
}
\If{$b\neq \emptyset$}{
Let $b'$ be the binary string that removes the last bit in $b$.\;
Find a $k$-covering set of $S_{b'}$ and let it be $S^k_{b'}$.\;
\For{$(x,y) \in S^k_{b'}\times S^k_{b}$}{
\If{$x,y\notin S^k_{b'} \cap S^k_b$}{
\texttt{IsShortcut}$(x,y,b') = \textsf{\em True}$.\;
 
Let $\widehat{d}_b(x,y) = d_b(x,y) + \mathcal{N}(0,\sigma^2)$.\;
}
}
}
}
5. \For{leaf node $(G_b,\text{-}) \in V(\mathcal{T})$}{
\For{$x,y\in V_b$ such that $x\neq y$}{
\texttt{IsShortcut}$(x,y,b) = \textsf{\em True}$.\;
 
Let $\widehat{d}_b(x,y) = d_b(x,y) + \mathcal{N}\left(0,\frac{2c^2\log(1.25/\delta')}{(\epsilon')^2}\right)$.\;
}
}
\Output{The binary tree $\mathcal{T}$, and $\widehat{d}_b(u,v)$ for all $u,v\in V$ and $b\in \{0,1\}^*$ such that $\texttt{IsShortcut}(x,y,b) = \textsf{\em True}$.}
\caption{Constructing private shortcuts via sub-sampling in separator.}\label{alg:shortcut_planar}
%\algolab{private_path}
\end{algorithm2e}
\clearpage

\end{document}